\documentclass[]{IEEEtran}
\usepackage{ifthen}
\usepackage{amsmath,amssymb,amsfonts,amsthm}
\usepackage[vlined,linesnumbered]{algorithm2e}
\usepackage{graphicx}
\usepackage{textcomp}
\usepackage[dvipsnames,rgb]{xcolor}
\usepackage{cancel}
\usepackage{bm}
\usepackage{hyperref}

\usepackage{tikz}
\usetikzlibrary{arrows, positioning, math, shadings, shapes.gates.logic.US, shapes.gates.logic.IEC, calc}

\usepackage{todonotes}

\usepackage{paralist}

\newcommand{\golden}[0]{C$^{\text{golden}}$}
\newcommand{\revised}[0]{C$^{\text{revised}}$}
\newcommand{\cequiv}{\mbox{\ensuremath{\bm{(\equiv)}}}}

\newcommand{\labelsec}[1]{\label{sec:#1}}
\newcommand{\refsec}[1]{Sect.~\ref{sec:#1}}
\newcommand{\refsecs}[2]{Sects.~\ref{sec:#1} and~\ref{sec:#2}}

\newtheorem{definition}{Definition}	
\newtheorem{lemma}{Lemma}
\newtheorem{theorem}{Theorem}
\newtheorem{example}{Example}

\newcommand{\labelfig}[1]{\label{fig:#1}}
\newcommand{\reffig}[1]{Fig.~\ref{fig:#1}}

\newcommand{\labelexample}[1]{\label{ex:#1}}
\newcommand{\refexample}[1]{Ex.~\ref{ex:#1}}

\newcommand{\labeltheorem}[1]{\label{thrm:#1}}
\newcommand{\reftheorem}[1]{Thrm.~\ref{thrm:#1}}

\newcommand{\labelalg}[1]{\label{alg:#1}}
\newcommand{\refalg}[1]{Alg.~\ref{alg:#1}}

\newcommand{\labeltab}[1]{\label{tab:#1}}
\newcommand{\reftab}[1]{Tab.~\ref{tab:#1}}

\newcommand{\Inv}[0]{\ensuremath\mathit{INV}}

\newcommand{\revisedff}[1]{\hat{#1}}

\newcommand{\mcl}[1]{\multicolumn{1}{c}{#1}}
\newcommand{\mclv}[1]{\multicolumn{1}{c|}{#1}}

\begin{document}

\title{Certified Sequential Sweep Without Unrolling}


\author{\IEEEauthorblockN{Tobias Seufert\hspace{5mm} Christoph Scholl}\\
	\IEEEauthorblockA{
		Department of Computer Science, University of Freiburg, Freiburg, Germany \hspace{5mm}\\
		$\{$seufert, scholl$\}$@cs.uni-freiburg.de\\
	}
}

\maketitle

\begin{abstract}
    Verifying retiming followed by additional sequential resynthesis steps remains challenging 
    for existing tools, limiting aggressive optimizations.
    Commonly-used equivalence checking tools rely on internal resynthesis operations and error-prone 
    orchestration of different model checkers, which makes certification difficult or even infeasible.
    We present an IC3-based technique that uses retiming as a preprocessing step and
    uses simulation to generate suspected invariants.
    Our technique efficiently verifies sequential 
    equivalence problems under retiming and arbitrarily strong sequential resynthesis,
    and has the additional feature of producing certificates all the same. 

    Our results on a selection of retimed and resynthesized open circuit designs show that 
    our rather simple approach vastly outperforms the whole portfolio of the winner of the latest Hardware Model Checking Competition as a representative of general-purpose certifying model checkers. 
    Compared to non-certifying approaches, like the mature equivalence checker of ABC, we are still
    competitive with additional complementary strengths.
\end{abstract}

\section{Introduction}
\labelsec{intro}

Retiming \cite{DBLP:journals/algorithmica/LeisersonS91} is among the most important logic optimization techniques used for
sequential circuits. By simply moving registers over combinational blocks, retiming enables reduction in either the delay
of the critical path or the required area of the design by reducing the number of required registers.

The verification of circuits which are optimized by \emph{only} retiming is likely to be simpler than general 
sequential equivalence checking. Technically, it is only necessary to establish the retiming invariant
 which represents the functional relation between the original and the revised design
after retiming---and prove it inductive \cite{mneimneh2003reverse}. Consider, for instance, an AND-gate in the original design that is driven by flip-flops (FFs) $a$ and $b$.
If in
the retimed design $a$ and $b$ are forward retimed and therefore moved over the AND-gate resulting in a new
flip-flop (FF) $c$, then a retiming invariant would state $c \equiv a \wedge b$.
The problem of checking, whether a retiming invariant (like $c \equiv a \wedge b$) is an 
inductive invariant, is coNP-complete \cite{DBLP:conf/iccad/JiangH07} (like 
combinational equivalence checking) instead of PSPACE-complete (like general sequential equivalence checking). 
If a design was \emph{only} retimed, we can either reconstruct the retiming invariant in polynomial time
\cite{mneimneh2003reverse} or use preprocessing techniques like maximal forward retiming (\refsec{max_forward}) which is then able to align the FFs.
If retiming was pre- or  succeeded by combinational resynthesis, these techniques can fail (\refsecs{retiming_invariants}{max_forward}). The problem remains likely harder than retiming alone \cite{DBLP:journals/tcad/JiangB06a},
but likely simpler than general sequential equivalence checking---because the next-state functions of the registers are preserved.
However, if sequential synthesis may change
next-state functions, the problem
becomes PSPACE-complete again. Moreover, an unknown number of alternating retiming
and combinational resynthesis steps is already PSPACE-complete \cite{DBLP:journals/tcad/JiangB06a}.

State-of-the-art approaches in sequential equivalence checking, such as ABC's
\texttt{dprove} \cite{DBLP:conf/cav/BraytonM10}, apply preprocessing steps such as maximal
forward retiming to align FFs between revised and original designs. 

However, from a certification perspective, it may be considered to be problematic to make use of
(uncertified) retiming steps to prove the correctness of (sequential logic synthesis
and) retiming. For maximal trust into the verification process, it is desirable 
that the verification engine produces a certificate that can be checked by an independent checker.

In this paper, we propose a \emph{certifying} method for sequential equivalence checking.
Our approach emits certificates under two preprocessing modes: maximal forward retiming
and targeted virtual forward retiming. 
Preprocessing steps like maximal forward retiming are handled by a certificate transformation
producing a certificate for the original problem before the preprocessing.
To the best of our knowledge, there is no other dedicated \emph{certifying} equivalence checker.
Compared to the best general purpose multi-engine model checker rIC3 \cite{DBLP:conf/cav/SuYCBH25}
(winner of two recent Hardware Model Checking Competitions), our approach performs significantly better
on sequential equivalence checking problems.
el checkers like IC3~\cite{DBLP:conf/vmcai/Bradley11}.

Our approach is built on IC3 and remains complete under
arbitrary orders of retiming and resynthesis transformations.
We show that IC3 is closely related in spirit to induction-based equivalence checking
approaches~\cite{DBLP:conf/date/Eijk98, DBLP:conf/fmcad/BjesseC00, DBLP:conf/iccad/JiangH07},
and leverage this connection by combining IC3 with speculated signal correspondences.
The proposed algorithm benefits from sequential-sweep ideas with speculative reduction
\cite{DBLP:conf/iccd/BaumgartnerMPKJ06,DBLP:conf/date/MonyBMB09,DBLP:journals/trets/MishchenkoBJJ11},
such as simulation-based refinement of equivalence classes, while retaining the
advantages of IC3 as a full model checker.
Beyond straightforward certification, we observe that our IC3-based framework 
often shows complementary
strengths compared to $k$-induction-based approaches that \emph{unroll} the transition relation.

\paragraph*{Related Work}

We consider our approach a natural extension of Van Eijk's work \cite{DBLP:conf/date/Eijk98}.
We do not require an exhaustive set of signal correspondences and can still prove
equivalence by automatically adding missing information.
Van Eijk's algorithm was followed up by extensions that incorporated stronger induction ($k$-induction) \cite{DBLP:conf/fmcad/SheeranSS00, DBLP:conf/fmcad/BjesseC00,DBLP:conf/iccad/JiangH07} because the large number of supporting signal correspondences, sometimes depending on virtual signals, 
were considered its main weakness. We argue that our approach complements $k$-induction-based approaches that depend on large unrollings: by using multi-timeframe rarity simulation \cite{DBLP:journals/trets/MishchenkoBJJ11}
we enable targeted identification of candidates for virtual forward retiming and 
additionally leverage the \emph{main strength of IC3} to quickly produce compact 
clausal reachability information to strengthen $k$-inductive signal correspondences into one-inductive ones.

Our approach of speculating invariants is related to research about speculative reduction \cite{DBLP:conf/dac/MonyBPK05, DBLP:conf/date/MonyBMB09, DBLP:conf/iccd/BaumgartnerMPKJ06, DBLP:conf/fmcad/DurejaBGKR24} of suspected corresponding signals. We implicitly prove signal correspondences by \emph{just assuming} that others hold. 

Since our speculated invariants contain internal signals---or even virtual signals that do not exist in the underlying circuit---our approach relates to IC3 with internal signals \cite{DBLP:conf/fmcad/DurejaGIV21}. Their approach is however not able to produce virtual signals if necessary or speculate.

The basic idea of strengthening IC3 with a set of suspected invariants is also applied 
in recent works making use of machine learned invariants \cite{DBLP:conf/aspdac/HuTYZZ24}. They refer to this as `side-loading'. However, compared to our approach, there is no attempt to distinguish between helpful and not helpful invariants. Furthermore, they are agnostic of internal signals.
Our implementation of \texttt{ic3-sim} (see \refalg{ic3_spec_invar} in \refsec{approach}) is related to \cite{DBLP:conf/fmcad/GurfinkelI15} and their concept of may-proof-obligations. We however do not only exclude the \emph{refuted} suspected invariants from propagation, but remove them completely from consideration. We also do not extract additional exact reachability information. 
Furthermore, 
\texttt{ic3-sim}
introduces restrictions like \cite{DBLP:conf/vmcai/SeufertSCRW22}, but with the ultimate goal of proving equivalence, not refuting it 
by detecting deep counterexample paths.

The authors of \cite{DBLP:conf/cav/YuBH20, DBLP:conf/fmcad/YuFBH22} propose an approach for certifying $k$-induction that could theoretically be used
to verify individual proofs of correspondence between two signals in existing equivalence checkers. It is, however, unclear how to obtain a monolithic,
independently checkable certificate for the entire equivalence-checking process. Note that a $k$-induction proof
for the entire sequential equivalence-checking problem did never succeed in our experiments (see \refsec{ret_res_pipe}).
Furthermore, their approach ignores simple-path (also called unique-state) constraints \cite{DBLP:conf/fmcad/SheeranSS00}.
We, however, observed that it is often necessary to consider these constraints in order to prove all
signal correspondences (see \refsec{scorr}).

Recently, certification of preprocessing techniques for hardware model checking has been investigated, for example temporal decompostion \cite{DBLP:conf/fmcad/YuFBH23} and phase abstraction \cite{DBLP:conf/ijcar/FroleyksYBH24}. 
To the best of our knowledge, retiming has not yet been covered as a preprocessing step.
\section{Preliminaries}
\labelsec{prelim}

\subsection{Synchronous Sequential Circuits}
\label{sse:sequentialcircuits}

We describe synchronous sequential circuits as six-tuples $C = (S,X,O,\Delta,\Lambda,i)$ with 
a set $S$ of state variables and a set $X$ of primary inputs. 
$O$ is a set of output variables representing outputs of $C$.
The set $\Delta$ contains for each variable $s \in S$ a Boolean formula $\delta_s$ over variables $S \cup X$
representing its next state function.
The set $\Lambda$ contains for each output variable $o \in O$ a Boolean formula $\lambda_o$ over $S \cup X$
representing the corresponding output function.
$i : S \rightarrow \{0, 1\}$ assigns an initial value $i(s)$ to each
state variable $s \in S$. 
$I(S) = \bigwedge_{s \in S} (s \equiv i(s))$ 
describes the initial state predicate characterizing all initial states that are allowed for $C$.

The transition relation $T(S, X, S')$ of $C$ is a Boolean formula that relates state variables from $S$ with the set of 
their \emph{next state copies} $S'$. $T$ is defined by $T = \bigwedge_{s \in S} (s' \equiv \delta_s)$.

Sequential circuits can be modeled as finite state machines where states are identified with Boolean assignments to variables of $S$.

A \emph{sequential miter} $M(C_1, C_2) = (S, X, O, \Delta, \Lambda, i)$ results from two circuits $C_1 = (S_1, X, O_1, \Delta_1, \Lambda_1, i_1)$ and $C_2 = (S_2, X, O_2, \Delta_2, \Lambda_2, i_2)$
with the same primary inputs and the same \emph{number of} pairwise corresponding outputs.
$S = S_1 \dot{\cup} S_2$, $\Delta = \Delta_1 \dot{\cup} \Delta_2$. 
$O$ contains a single output variable $o_M$,  
$\Lambda$ contains a Boolean formula $\lambda_{o_M}(S, X)$ representing a Boolean
function that outputs a 0 for an assignment to $(S, X)$,
if \emph{corresponding} outputs of $C_1$ and $C_2$ 
evaluate to the same value,
and a 1 otherwise.
The FFs in $M(C_1, C_2)$ are initialized as in $C_1$ and $C_2$, i.e.,
$i(s) = i_1(s)$, if $s \in S_1$, and $i(s) = i_2(s)$, if $s \in S_2$.

We consider $C_1$ and $C_2$ to be \emph{equivalent}, if the output $o_M$ of $M(C_1, C_2)$
remains at logic 0 under all possible input sequences.
The Boolean predicate $\cequiv(S, X)$ represents the formula 
$\lambda_{o_M} \equiv 0$.

\subsection{Retiming}

Retiming \cite{DBLP:journals/algorithmica/LeisersonS91} is a technique for sequential resynthesis
with the goal
of either reducing the number of FFs or the delay of the critical path. 
Retiming transforms a specification circuit \golden{} into an implementation circuit \revised{} by \emph{moving} FFs 
over logic gates or fan-outs.

An FF can be moved from the output of a gate $f$ to \emph{all} of its inputs.
This is called \emph{backward retiming}.
FFs can also be moved from the inputs (removing one FF from each input) to the output of
a gate $f$. This is called \emph{forward retiming}.
Similarly, backward retiming can remove one FF from each branch of a fan-out and add an FF to the fan-out stem.
Forward retiming can move an FF from a fan-out stem to all of its branches.
In the following we denote FFs from \revised{} with an accent like $\revisedff{s}_i$.

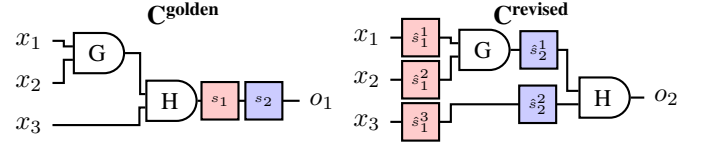
\begin{figure}
    \tikzmath{\scale = 0.9; \fdist = \scale * 1; \fffdist = \fdist * 1.5; \outdist = \fdist * 0.2; \idist = \fdist / 10; \ffwidth = \scale * 2; \ffheight = \scale * 5.5; \idisthalf = \idist / 2; \idisttwo = \idist * 2; \idistthree = \idist * 3;}%
\centering
\begin{tabular}{cc}
\textbf{\golden{}} & \textbf{\revised{}} \\[-0.6em]
\hspace{-1em}
\begin{tikzpicture}[>=latex, thick]
  \node (i1) {$x_1$};
  \node[below=\idist cm of i1] (i2) {$x_2$};
  \node[below=\idist cm of i2] (i3) {$x_3$};

  \node[and gate US, draw, logic gate inputs=nn, right=\idistthree cm of i1, scale=\scale, yshift=-\idisttwo cm] (f1) {G};
  \draw (i1.east) -- ++(0.15,0) |- (f1.input 1);
  \draw (i2.east) -- ++(0.15,0) |- (f1.input 2);

  \node[and gate US, draw, logic gate inputs=nn, right=\idistthree cm of f1, scale=\scale, yshift=-0.7cm] (f2) {H};
  \draw (f1.output) -- ++(0.2,0) |- (f2.input 1);
  \draw (i3.east) -- ++(1.15,0) |- (f2.input 2);

  \node[right=\fffdist cm of f2] (o) {$o_1$};
  \draw (f2) -- node[pos=0.22, draw, rectangle, minimum width=\ffwidth mm, minimum height=\ffheight mm, fill=red!20]{{\tiny $s_1$}} 
                node[pos=0.65, draw, rectangle, minimum width=\ffwidth mm, minimum height=\ffheight mm, fill=blue!20]{{\tiny $s_2$}}(o);
\end{tikzpicture}
&\hspace{-1.5em}
\begin{tikzpicture}[>=latex, thick]
  \node (i1) {$x_1$};
  \node[below=\idist cm of i1] (i2) {$x_2$};
  \node[below=\idist cm of i2] (i3) {$x_3$};

  \node[and gate US, draw, logic gate inputs=nn, right=\fdist cm of i1, scale=\scale, yshift=-\idisttwo cm] (f1) {G};
  \draw (i1.east) -- node[pos=0.5, draw, rectangle, minimum width=\ffwidth mm, minimum height=\ffheight mm, fill=red!20]{\tiny{$\revisedff{s}^1_{1}$}} ++(0.8,0) |- (f1.input 1);
  \draw (i2.east) -- node[pos=0.5, draw, rectangle, minimum width=\ffwidth mm, minimum height=\ffheight mm, fill=red!20]{\tiny{$\revisedff{s}^2_{1}$}} ++(0.8,0) |- (f1.input 2);
  \node[and gate US, draw, logic gate inputs=nn, right=\fdist cm of f1, scale=\scale, yshift=-0.7cm] (f2) {H};
  \draw (f1.output) -- node[pos=0.5, draw, rectangle, minimum width=\ffwidth mm, minimum height=\ffheight mm, fill=blue!20]{\tiny{$\revisedff{s}^1_{2}$}}  ++(0.75,0) |- (f2.input 1);
  \draw (i3.east) -- node[pos=0.5, draw, rectangle, minimum width=\ffwidth mm, minimum height=\ffheight mm, fill=red!20]{\tiny{$\revisedff{s}^3_{1}$}} 
                 ++(0.8,0) |- node[pos=0.84, draw, rectangle, minimum width=\ffwidth mm, minimum height=\ffheight mm, fill=blue!20]{\tiny{$\revisedff{s}^2_{2}$}}(f2.input 2);

  \node[right=\outdist cm of f2] (o) {$o_2$};
  \draw (f2) -- (o);
\end{tikzpicture}
\end{tabular}
\vspace{-1.0em}
    \caption{A simple backward retimed circuit.\labelfig{retiming}}
    \vspace*{-5mm}
\end{figure}
\begin{example}\labelexample{retiming}
    We consider the two circuits depicted in \reffig{retiming}.
    Circuit \golden{} has three primary inputs $x_1,x_2,x_3$, two FFs $s_1, s_2$, and an output $o_1$.
    The output is computed by the FF output of $s_2$, whereas the next state functions of $s_1, s_2$ are defined as
    $\delta_{s_1} = (x_1 \wedge x_2) \wedge x_3$
    and $\delta_{s_2} = s_1$.
    The circuit \revised{} has the same interface, but FF $s_1$ in \golden{} is backward retimed 
    to 
    $\revisedff{s}^1_{1},\revisedff{s}^2_{1},\revisedff{s}^3_{1}$
    in \revised{} over both AND gates $G$ and $H$, FF $s_2$ is backward retimed to 
    $\revisedff{s}^1_{2},\revisedff{s}^2_{2}$ only over AND gate $H$.
\end{example}

Retiming adjusts the initial states as well. 
Assume that $i(s_1) = \epsilon_1$ and $i(s_2) = \epsilon_2$ in \refexample{retiming}.
Then $i(\revisedff{s}_{21})$ and $i(\revisedff{s}_{22})$ are chosen in a way that 
$i(\revisedff{s}^1_{2}) \wedge i(\revisedff{s}^2_{2}) = \epsilon_2$.
$i(\revisedff{s}^1_{1})$, $i(\revisedff{s}^2_{1})$, and $i(\revisedff{s}^3_{1})$ are chosen such that
$(i(\revisedff{s}^1_{1}) \wedge i(\revisedff{s}^2_{1})) \wedge i(\revisedff{s}^3_{1}) = \epsilon_1$.
Finding a ``correct initialization'' is not a problem for forward retimings, since circuits
implement Boolean \emph{functions} and a correct initialization of the implementation circuit
can be computed by applying Boolean functions.
It is easy to see that a correct initialization does not exist for all \emph{backward} retimings and
all initializations of the FFs in \golden{}. 
\section{Techniques used in Retiming Verification}
\labelsec{approach}

In this section we give a brief review of techniques that inspired our approach to verifying
and at the same time certifying the correctness of retiming and synthesis techniques.

\subsection{Retiming Invariants}\labelsec{retiming_invariants}

One approach to verify retiming is the construction of a so-called retiming invariant \cite{mneimneh2003reverse}.
Retiming invariants are a special case of inductive invariants.

One possibility to prove that some predicate
$P(S, X)$ (like $\cequiv(S, X)$ expressing
equivalence of two sequential circuits as described in Sect.~\ref{sse:sequentialcircuits}) 
holds in all reachable states of a sequential circuit $C$ (like a miter circuit according
to Sect.~\ref{sse:sequentialcircuits}), i.e., to prove that $P(S, X)$ is an \emph{invariant} of $C$,
is to compute a so-called \emph{inductive} invariant
$INV(S)$. $INV(S)$ is called an inductive invariant, if the following holds:
\begin{align}
   &I(S) \implies\Inv(S) \label{inv1} \\
   &\Inv(S) \implies P(S, X) \label{inv2} \\
   &\Inv(S) \wedge T(S,X,S') \implies \Inv(S') \label{inv3}
\end{align}
Condition (\ref{inv1}) requires that $INV(S)$ holds in all
initial states. Condition (\ref{inv3}) requires that $INV(S)$ is
\emph{inductive}, i.e., it holds in the next state, if it holds in the current state.
Condition (\ref{inv1}) and (\ref{inv3}) imply that $INV$ is an invariant. 
Condition (\ref{inv2}) says that 
$P(S, X)$ holds in all states satisfying the inductive invariant $INV(S)$,
i.e., it is an invariant as well.

The computation of a retiming invariant relies on the observation that
retiming can be decomposed into a series of basic FF movements
(backward and forward movements of FF over gates and fanouts, see Sect. II.B).
The work in \cite{mneimneh2003reverse} shows how to reconstruct this series of FF movements
just from a pair of circuits \golden{} and \revised{}.
The retiming invariant can be computed step by step from this series of FF movements.
For more details see \refsec{retiming_preproc}.

The advantage of the approach using retiming invariants is that the retiming invariant
is not destroyed, if we apply combinational logic synthesis after retiming.
(This is immediately clear, since combinational logic synthesis may change
the representations of transition and output functions, but not their semantics.)
The disadvantage of the approach lies in the fact that it is only applicable
to retiming alone. If combinational or even sequential logic synthesis has already been
applied, the reconstruction of the original retiming FF movements does not work anymore.
Moreover, sequential logic synthesis may destroy an existing retiming invariant.

\subsection{Maximal Forward Retiming}\labelsec{max_forward}

Another option that is used in approaches like \texttt{dprove} in ABC \cite{DBLP:conf/cav/BraytonM10}
is Maximal Forward Retiming that moves all FFs as far
as possible into the direction of the primary outputs, both in \golden{} and \revised{}.

Maximal Forward Retiming is a successful verification approach, especially if only retiming
has been used without any combinational or sequential logic synthesis.
The following example shows however that its success may already be destroyed
by weak combinational logic synthesis operations:

\begin{figure}
    \tikzmath{\scale = 0.9; \fdist = \scale * 0.9; \ffdist = \fdist * 1.5; \fffdist = \fdist * 1.7; \ffffdist = \fdist * 1.9;\outdist = \fdist * 0.2; \idist = \fdist / 5; \ffwidth = \scale * 1.5; \ffheight = \scale * 5.5; \idisthalf = \idist / 2; \idistthird = \idist / 3; \idisttwo = \idist * 2;}%
\centering
\begin{tabular}{cc}
\hspace{2cm}\textbf{\golden{}} & \hspace{2cm}\textbf{\revised{}} \\[-0.8em]
\hspace{-1em}
\begin{tikzpicture}[>=latex, thick]
  \node (i0) {$x_0$};
  \node[below=\idist cm of i0] (i1) {$x_1$};
  \node[below=\idist cm of i1] (i2) {$x_2$};
  \node[below=\idist cm of i2] (i3) {$x_3$};

  \node[and gate US, draw, logic gate inputs=nn, right=\idisttwo cm of i0, scale=\scale, yshift=-\idist cm] (f0) {G};
  \draw (i0.east) -- ++(0.15,0) |- (f0.input 1);
  \draw (i1.east) -- ++(0.15,0) |- (f0.input 2);

  \node[and gate US, draw, logic gate inputs=nn, right=\fffdist cm of f0, scale=\scale, yshift=-\idisttwo cm] (f1) {H};
  \draw (f0.output) -- node[pos=0.25, draw, rectangle, minimum width=\ffwidth mm, minimum height=\ffheight mm, fill=red!20]{{\tiny $s_1$}} 
                node[pos=0.75, draw, rectangle, minimum width=\ffwidth mm, minimum height=\ffheight mm, fill=blue!20]{{\tiny $s_2$}} ++(\ffdist,0) |- (f1.input 1);
  \draw (i2.east) -- node[pos=0.32, draw, rectangle, minimum width=\ffwidth mm, minimum height=\ffheight mm, fill=red!20]{{\tiny $s_1$}} 
                node[pos=0.75, draw, rectangle, minimum width=\ffwidth mm, minimum height=\ffheight mm, fill=blue!20]{{\tiny $s_2$}} ++(\fffdist,0) |- (f1.input 2);

  \node[and gate US, draw, logic gate inputs=nn, right=\idist cm of f1, scale=\scale, yshift=-0.7cm] (f2) {I};
  \draw (f1.output) -- ++(0.1,0) |- (f2.input 1);
  \draw (i3.east) -- ++(\ffffdist,0) |- (f2.input 2);

  \node[right=\outdist cm of f2] (o) {};
  \draw (f2) -- (o);
\end{tikzpicture}
&\hspace{-1.5em}
\begin{tikzpicture}[>=latex, thick]
  \node (i0) {$x_0$};
  \node[below=\idist cm of i0] (i1) {$x_1$};
  \node[below=\idist cm of i1] (i2) {$x_2$};
  \node[below=\idist cm of i2] (i3) {$x_3$};

  \node[and gate US, draw, logic gate inputs=nn, right=\ffffdist cm of i0, scale=\scale, yshift=-\idist cm] (f0) {G};
  \draw (i0.east) -- node[pos=0.3, draw, rectangle, minimum width=\ffwidth mm, minimum height=\ffheight mm, fill=red!20]{{\tiny $\revisedff{s}^1_{1}$}} 
                node[pos=0.75, draw, rectangle, minimum width=\ffwidth mm, minimum height=\ffheight mm, fill=blue!20]{{\tiny $\revisedff{s}^1_{2}$}} ++(\fffdist,0) |- (f0.input 1);
  \draw (i1.east) -- node[pos=0.3, draw, rectangle, minimum width=\ffwidth mm, minimum height=\ffheight mm, fill=red!20]{{\tiny $\revisedff{s}^2_{1}$}} 
                node[pos=0.75, draw, rectangle, minimum width=\ffwidth mm, minimum height=\ffheight mm, fill=blue!20]{{\tiny $\revisedff{s}^2_{2}$}} ++(\fffdist,0) |- (f0.input 2);

  \node[and gate US, draw, logic gate inputs=nn, right=\ffffdist cm of i2, scale=\scale, yshift=-\idisttwo cm] (f1) {J};
  \draw (i2.east) -- node[pos=0.3, draw, rectangle, minimum width=\ffwidth mm, minimum height=\ffheight mm, fill=red!20]{{\tiny $\revisedff{s}_1$}} 
                node[pos=0.75, draw, rectangle, minimum width=\ffwidth mm, minimum height=\ffheight mm, fill=blue!20]{{\tiny $\revisedff{s}_2$}} ++(\fffdist,0) |- (f1.input 1);
  \draw (i3.east) -- ++(\fffdist,0) |- (f1.input 2);
  \node[and gate US, draw, logic gate inputs=nn, right=2.5 cm of i1, scale=\scale, yshift=-\idist cm] (f2) {K};
  \draw (f0.output) -- ++(\idisthalf,0) |- (f2.input 1);
  \draw (f1.output) -- ++(\idisthalf,0) |- (f2.input 2);

  \node[right=\outdist cm of f2] (o) {};
  \draw (f2) -- (o);
\end{tikzpicture}
\end{tabular}
\vspace{-1.5em}
    \caption{A backward retimed and resynthesized circuit.\labelfig{retiming_virt}}
    \vspace*{-5mm}
\end{figure}
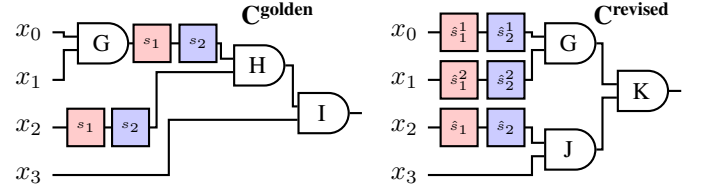

\begin{example}\labelexample{retiming_most_forward}
We consider the two \emph{sequentially equivalent} circuits in \reffig{retiming_virt}. The 
circuit \revised{} results from \golden{} by  
retiming followed by combinational resynthesis.
After backward retiming of FFs $s_1$ and $s_2$, gates $H, I$ are replaced by
$J, K$.
It is easy to see that, without the resynthesis step, Maximal Forward Retiming 
would transform both \golden{} and \revised{} into identical circuits with all 
FFs moved to the line that is driven by the AND-gate $H$.
The resynthesis step of rewriting $H, I$ to $J, K$, though, hinders Maximal Forward Retiming 
and it will not be able to align the FFs.
\end{example}

This observation explains that Maximal Forward Retiming (together with other
retiming operations like retiming for minimizing the number of FFs) is only used
as a preprocessing step in approaches like \texttt{dprove} in ABC \cite{DBLP:conf/cav/BraytonM10} to simplify
the equivalence checking problem. This preprocessing step may or may not help for the
overall verification task.

\subsection{Signal Correspondences}\labelsec{scorr}

As already mentioned in \refsec{retiming_invariants}, combinational and / or sequential logic synthesis
may destroy the possibility of computing a retiming invariant.
Therefore approaches like \cite{DBLP:conf/date/Eijk98,DBLP:conf/fmcad/BjesseC00, DBLP:conf/iccad/JiangH07} compute signal correspondences between signals
in \golden{} and \revised{} instead, in order to simplify the verification task.
There is a signal correspondence between two signals in a sequential circuit, if the two signals
always assume the same values in all runs of the sequential circuit starting with the initial state.
Candidates for signal correspondence can be computed by sequential simulation,
the proof of signal correspondences has to be performed by more powerful techniques
like $k$-induction. 
Signal correspondences are typically computed step by step based on already
proven signal correspondences.
Two corresponding signals are replaced by one representative.
It is clear that the verification task has been solved as soon as 
signal correspondences have been proved for all corresponding outputs
of \golden{} and \revised{}.

Our approach is inspired by the computation of signal correspondences, but
it uses speculated signal correspondences as a means to construct an inductive
invariant in a modified IC3 framework.

\subsection{Virtual Forward Retiming}\labelsec{virt_retiming}

Virtual Forward Retiming has been introduced in the work of Van Ejik \cite{DBLP:conf/date/Eijk98} in the context of sequential equivalence checking of retimed circuits.
In our approach we use Virtual Forward Retiming to improve the computation of
(speculated) signal correspondences (for computing speculated invariants).

We recall Virtual Forward Retiming informally.
For a gate $G$ whose inputs are driven by FFs
$s_1,\dots,s_m$ with input signals $t_1,\dots,t_m$, respectively,
we define by ``Virtual Forward Retiming'' a virtual signal $l\equiv G(t_1,\dots,t_m)$. 
The virtual signal is left
dangling, therefore this operation leaves the transition system unchanged.
Virtual Forward Retiming only augments the set of candidate signals for
signal correspondences. The additional signal correspondences
can be used for constructing inductive invariants.
(If all $t_j$ are also driven by FFs, we can also retime them over $G$. 
In that way, Virtual Forward Retiming can also retime two or more \emph{layers} of FFs.)

The following example shows that even in the case of pure retiming without logic synthesis
the computation of signal correspondences is not enough for obtaining an \emph{inductive} invariant, but
Virtual Forward Retiming is needed to find an inductive invariant based on signal correspondences:

\begin{figure}
    \tikzmath{\scale = 0.9; \fdist = \scale * 1; \fffdist = \fdist * 1.5; \outdist = \fdist * 0.2; \idist = \fdist / 10; \ffwidth = \scale * 2; \ffheight = \scale * 5.5; \idisthalf = \idist / 2; \idisttwo = \idist * 2; \idistthree = \idist * 3;}%
\centering
\begin{tabular}{cc}
\textbf{\golden{}} & \textbf{\revised{}} \\[-0.8em]
\hspace{-1em}
\begin{tikzpicture}[>=latex, thick]
  \node (i1) {$x_1$};
  \node[below=\idist cm of i1] (i2) {$x_2$};
  \node[below=\idist cm of i2] (i3) {$x_3$};

  \node[and gate US, draw, logic gate inputs=nn, right=\idistthree cm of i1, scale=\scale, yshift=-\idisttwo cm] (f1) {G};
  \draw (i1.east) -- ++(0.15,0) |- (f1.input 1);
  \draw (i2.east) -- ++(0.15,0) |- (f1.input 2);

  \node[and gate US, draw, logic gate inputs=nn, right=\idistthree cm of f1, scale=\scale, yshift=-0.7cm] (f2) {H};
  \draw (f1.output) -- ++(0.2,0) |- (f2.input 1);
  \draw (i3.east) -- ++(1.15,0) |- (f2.input 2);

  \node[right=\fffdist cm of f2] (o) {$o_1$};
  \draw (f2) -- node[pos=0.22, draw, rectangle, minimum width=\ffwidth mm, minimum height=\ffheight mm, fill=red!20]{{\tiny $s_1$}} 
                node[pos=0.65, draw, rectangle, minimum width=\ffwidth mm, minimum height=\ffheight mm, fill=blue!20]{{\tiny $s_2$}}(o);
\end{tikzpicture}
&\hspace{-1.5em}


\begin{tikzpicture}[>=latex, thick]
  \node (i1) {$x_1$};
  \node[below=\idist cm of i1] (i2) {$x_2$};
  \node[below=\idist cm of i2] (i3) {$x_3$};

  \node[and gate US, draw, logic gate inputs=nn, right=\fdist cm of i1, scale=\scale, yshift=-\idisttwo cm] (f1) {G};
  \draw (i1.east) -- node[pos=0.5, draw, rectangle, minimum width=\ffwidth mm, minimum height=\ffheight mm, fill=red!20]{\tiny{$\revisedff{s}^1_{1}$}} ++(0.8,0) |- (f1.input 1);
  \draw (i2.east) -- node[pos=0.5, draw, rectangle, minimum width=\ffwidth mm, minimum height=\ffheight mm, fill=red!20]{\tiny{$\revisedff{s}^2_{1}$}} ++(0.8,0) |- (f1.input 2);

  \node[and gate US, draw, logic gate inputs=nn, right=\fdist cm of f1, scale=\scale, yshift=-0.7cm] (f2) {H};
  \draw (f1.output) -- ++(0.1,0) coordinate (preff12)
                    -- node[pos=0.5, draw, rectangle, minimum width=\ffwidth mm, minimum height=\ffheight mm, fill=blue!20]{\tiny{$\revisedff{s}^1_{2}$}} ++(0.7,0) |- (f2.input 1);
  \draw (i3.east) -- node[pos=0.5, draw, rectangle, minimum width=\ffwidth mm, minimum height=\ffheight mm, fill=red!20]{\tiny{$\revisedff{s}^3_{1}$}} ++(0.8,0) coordinate (preff22)
                 |- node[pos=0.87, draw, rectangle, minimum width=\ffwidth mm, minimum height=\ffheight mm, fill=blue!20]{\tiny{$\revisedff{s}^2_{2}$}} (f2.input 2);

  \fill (preff12) circle (1.5pt);
  \fill (preff22) circle (1.5pt);

  \node[and gate US, draw, dash pattern=on 1pt off 1pt, logic gate inputs=nn, below=0.3cm of f2, scale=\scale] (f2v) {H'};
  \draw[dash pattern=on 1pt off 1pt] (preff12) |- (f2v.input 1);
  \draw[dash pattern=on 1pt off 1pt] (preff22) |- (f2v.input 2);

  \node[right=\outdist cm of f2] (o2) {$o_2$};
  \draw (f2) -- (o2);
  \node[right=\outdist cm of f2v] (o3) {$l$};
  \draw[dash pattern=on 1pt off 1pt] (f2v) -- (o3);
\end{tikzpicture}
\end{tabular}
\vspace{-1.0em}
    \caption{A backward retimed circuit with virtual forward retiming.\labelfig{retiming_gates_virt}}
    \vspace*{-5mm}
\end{figure}
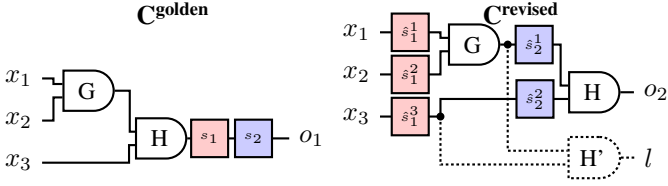
\begin{example}\labelexample{retiming_gates_virt}
    We revisit \refexample{retiming} again and assume that all FFs are initialized to 0.
	To establish equivalence would mean that we have to prove that $s_2 \equiv \revisedff{s}^1_{2} \wedge \revisedff{s}^2_{2}$.
    The constraint $s_2 \equiv \revisedff{s}^1_{2} \wedge \revisedff{s}^2_{2}$ alone does not form
    an inductive invariant.
    A simple counterexample to induction would be that $\delta_{s_2} = s_1 = 0$
	and  $\delta_{\revisedff{s}^1_{2}} \wedge \delta_{\revisedff{s}^2_{2}} = (\revisedff{s}^1_{1} \wedge \revisedff{s}^2_{1}) \wedge \revisedff{s}^3_{1} = 1$.
	Obviously, this counterexample is unreachable from the initial all-zero state. 
   It is easy to see that no other signal correspondences exist between \golden{} and \revised{}.
However, if we add 
the signal $l \equiv (\revisedff{s}^1_{1} \wedge \revisedff{s}^2_{1}) \wedge \revisedff{s}^3_{1}$ by \emph{virtually} forward retiming $\revisedff{s}^1_{2}$ and $\revisedff{s}^2_{2}$ over
$H$, we get another signal correspondence $s_1 \equiv l$. It is easy to see that together
$(s_2 \equiv \revisedff{s}^1_{2} \wedge \revisedff{s}^2_{2}) \wedge (s_1 \equiv (\revisedff{s}^1_{1} \wedge \revisedff{s}^2_{1}) \wedge \revisedff{s}^3_{1})$ become inductive. We 
illustrate this transformation in \reffig{retiming_gates_virt}.
\end{example}

This observation motivated us to include Virtual Forward Retiming into our framework
using speculated invariants.

\subsection{IC3}
\labelsec{ic3pdr}

IC3~\cite{DBLP:conf/vmcai/Bradley11,DBLP:conf/fmcad/EenMB11} is a general purpose method to prove invariant properties $P(S,X)$, like $\cequiv(S,X)$, of a system by finding an inductive invariant $\Inv(S)$ for which it holds that $\Inv(S) \Rightarrow P(S,X)$. 

In our work we present a modified IC3 framework that is tailored towards sequential equivalence checking
problems that may result from retiming and logic synthesis.

IC3 incrementally constructs overapproximations, called the trace, $F_0, F_1,\ldots,F_n$ of the reachable states in up to $i \in \{1,\ldots,n\}$ steps---with one exception: $F_0(S) = I(S)$ is represented exactly. The $F_i$ are gradually strengthened by removing provably unreachable states (starting from $I$) that contain predecessors of the \emph{bad} states (states that violate $P(S,X)$ under at least one input assignment on $X$ \emph{as well as their predecessors}). 
A \emph{bad} state $c$ in time frame $i$ -- such a state is called a \emph{proof obligation} (POB) --
has to be proven unreachable within up to $i$ steps, because otherwise there would be a counterexample to $P(S,X)$ being invariant. The first candidate in a tree of POBs is always a state satisfying $\neg P(S,X)$, extracted as a satisfying assignment from $F_n \wedge \neg P(S,X)$.

\subsubsection{\textbf{ResolveObligations()}}\labelsec{resolve_obl}
POBs are resolved in a recursive manner by a procedure we name \emph{ResolveObligations()}. \emph{ResolveObligations(($F_0, F_1,\ldots,F_n$), $c(S)$, $i$)} operates on the current trace of all time frames $F_0, F_1,\ldots,F_n$ and tries to resolve a POB $c$ in time frame $F_i$.
It queries a SAT solver with $\neg c(S) \wedge F_{i-1}(S) \wedge T(S, X, S') \wedge c(S')$. If this is unsatisfiable
then $c$ is resolved and removed from $F_i$ (because it is proven unreachable within up to $i$ steps) by 
adding resp. \emph{learning} the clause $\neg c$ (which can also be generalized) in all time frames $F_k$ with $k \in \{1,\ldots, i\}$.

If $\neg c(S) \wedge F_{i-1}(S) \wedge T(S, X, S') \wedge c(S')$ is satisfiable, the satisfying assignment $d$ on $S$ forms a new POB in frame $i-1$, and \emph{ResolveObligations()} recurses on \emph{ResolveObligations(($F_0, F_1,\ldots,F_n$), $d(S)$, $i-1$)}, and so on. If a POB in $F_0$ is encountered, then \emph{ResolveObligations()} reports a counterexample ($s, <x_1,\ldots,x_n>)$\footnote{Note that a simple and common optimization enables IC3 to report counterexamples longer than the trace. For the sake of simplicity, though, we assume that this optimization is not in place.} with 
an initial state $s(S) \implies F_0(S)$ and an input sequence $<x_1,\ldots,x_n>$ to reach $\neg P(S,X)$. 
Then IC3 reports `\emph{UNSAFE}'.

\subsubsection{\textbf{Propagate()}}\labelsec{propagate}
Once all POBs are resolved in the current trace $F_0, F_1,\ldots,F_n$ (and therefore $F_n \wedge \neg P(S,X)$ is unsatisfiable), the trace is increased by one new time frame $F_{n+1} = \top$ into $F_0, F_1,\ldots,F_n, F_{n+1}$
and the procedure \emph{Propagate()} is called. 
\emph{Propagate($F_1,\ldots,F_n, F_{n+1}$)} tries to propagate \emph{all} learned clauses into later 
time frames. Thus, for each pair of time frames $F_i, F_{i+1}$ with $i \in \{1,\ldots,n\}$, it tries to propagate all clauses $\neg c \in (F_i \setminus F_{i+1})$ from $F_i$ to $F_{i+1}$ via querying a SAT solver with $F_i(S) \wedge T(S,X,S') \wedge c(S')$ (propagation is successful, if this is unsatisfiable).

If at some point during \emph{Propagate()} $F_i \equiv F_{i+1}$, then \emph{Propagate()} and therefore also IC3 reports `\emph{SAFE}'.

Following these principles, IC3 maintains the following invariants: (1) $I \implies F_i$ for all $i \in \{0,\ldots,n\}$, (2) $F_i(S) \implies P(S,X)$ for all $i \in \{0,\ldots,n-1\}$, (3) $F_i(S) \implies F_{i+1}(S)$ for all $i \in \{0,\ldots,n-1\}$, and (4) $F_i(S) \wedge T(S,X,S') \implies F_{i+1}(S')$ for all $i \in \{0,\ldots,n-1\}$. 

\section{Our Approach}

Our approach is designated to verify and end-to-end certify miters between
an original sequential circuit and a version that results from retiming and resynthesis.
It works with speculated invariants and integrates the computation of signal correspondences by a ``sequential sweep''
(\refsec{scorr}) into IC3 (\refsec{ic3pdr}). 
We strengthen the sequential sweep by targeted virtual forward retiming (\refsec{virt_retiming}).
Moreover, we use 
\emph{certified} preprocessing by maximal forward retiming (\refsec{retiming_preproc}).

We speculate FF-output-to-signal (or delayed signal, \refsec{sse:sequentialsimulation})
correspondences in equivalence classes.
We extend IC3 by a simulation engine for equivalence class refinement while retaining its 
capabilities as a complete model checking algorithm (\refsec{ssweep_ic3}).

We call our approach \texttt{ic3-sim} if only virtual forward retiming is performed lazily
and \texttt{ic3-sim-F} if most-forward retiming is performed upfront.

\subsection{Sequential Simulation}
\labelsec{sse:sequentialsimulation}

As already mentioned in \refsec{scorr} we use sequential simulation to obtain candidate pairs
of equivalent signals
to be used as speculated invariants in our framework.
By simulating the miter similar to \cite{mishchenko2012using}, we compute equivalence classes
of potentially equivalent signals. 

If we want to find an inductive signal correspondence without unrolling, it can be 
necessary to produce virtual signals by \emph{virtual forward retiming} (see \refsec{virt_retiming}).
In order to find candidates, we need a means to simulate the circuit over a 
multi-timeframe window. This window is bounded by what is often referred to as
\emph{register depth} $d$ of the circuit \cite{DBLP:conf/iccad/JiangH07}. 

\begin{definition}[Register Depth]
    The \emph{register depth} $d$ of a synchronous sequential circuit $C$ is the maximum
    number of FFs on any simple directed path from a primary input to a primary
    output in $C$, where a simple path visits no node more than once.
\end{definition}

During simulation of the current time frame, we keep the signal values of the past 
$d-1$ frames. 
If, e.g., a FF $s$ always holds the same value that a line $l$ holds $n$ time frames later (with $0 <n < d$),
we say that $s$ with a \emph{delay} of zero is equivalent 
to $l$ with a \emph{delay} of $n$.
$l$ is then a candidate for virtual retiming of $n$ layers of FFs driving
it, like described in \refsec{virt_retiming}.

Another viable (and complete) option is to ignore delayed correspondences and let
IC3 add the missing parts. 
Preliminary results showed that this is sometimes even beneficial,
but in most cases not.

\subsection{A Simple Check For Suspected Invariants}
Here we describe a simple check whether a suspected invariant $\Inv$ is really
inductive. Later on in \refsec{ssweep_ic3}, we will show that IC3 can be used to 
effectively simulate this simple check.
Moreover, we extend IC3 for computing a real inductive invariant
for the case that the speculated invariant is not a real inductive invariant.

The suspected invariant $\Inv$ will most likely be provided to us by 
FF-output-to-signal (or virtual signal) correspondences found by sequential simulation.
Another option could be hints emitted by the synthesis tool that applied the 
transformation steps. 
Our approach can make use of speculated invariants from any source.

An option to prove $\Inv(S) \wedge T(S,X,S') \implies \Inv(S')$ from Condition~(\ref{inv3}) is to iteratively prove that \emph{for all}
conjuncts $r_i$ of $\Inv$, it holds that
	$\Inv(S) \wedge T(S, X, S') \implies r_i(S')$.

In \refalg{retiming_invar} we describe a SAT-based procedure that decides whether a given suspected invariant $\Inv$ is indeed a proof of sequential equivalence.

\begin{algorithm}[tb]
\begin{footnotesize}
\ForEach{conjunct $r_i$ in $\Inv$}{
	\lIf{SAT? [ $I(S) \wedge \neg r_i(S)$ ]}{
		\Return \emph{FAIL}
	}
}
\ForEach{conjunct $r_i$ in $\Inv$}{\lIf{SAT? [ $\Inv(S) \wedge T(S, X, S') \wedge \neg r_i(S')$ ]}{\Return \emph{FAIL}}}
\leIf{SAT? [ $\Inv(S) \wedge \neg \cequiv(S,X)$ ]}{\Return \emph{FAIL}\;}{	\Return \emph{SUCCESS} }
\end{footnotesize}
\caption{Simple induction check.\labelalg{retiming_invar}}
\end{algorithm}

\subsection{IC3 Sequential Sweep}\labelsec{ssweep_ic3}
\paragraph{\textbf{\texttt{ic3-load-$F_1$}}}
Our first proposal is to use IC3 to verify sequential equivalence and
to initialize the first frame $F_1$ with information extracted from a
suspected retiming invariant $\Inv$.
We describe our modification of IC3 in \refalg{ic3_retiming_invar}
which we call \texttt{ic3-load-$F_1$}.
We use the ideas of \texttt{ic3-load-$F_1$} to prepare the introduction of
\refalg{ic3_spec_invar} later on. 

\begin{algorithm}[b]
\begin{footnotesize}
\ForEach{conjunct $r_i$ in $\Inv$}{
	\lIf{SAT? [ $I(S) \wedge \neg r_i(S)$ ]}{
		\label{line:i_implies_R}
		$\Inv = \Inv \setminus \{r_i\}$
	}
}
\ForEach{conjunct $r_i$ in $\Inv$}{
	\lIf{SAT? [ $I(S) \wedge T(S,X, S') \wedge \neg r_i(S')$ ]}{
		\label{line:i_and_T_implies_R}
		$\Inv = \Inv \setminus \{r_i\}$
	}
}
\leIf{IC3$_{[F_1:=\Inv]} (I,T, \cequiv)$ == \emph{SAFE}}{
	\label{line:F1_equals_R}
	\Return \emph{EQUIVALENT}\;
}{
	\Return \emph{INEQUIVALENT}
}
\end{footnotesize}
\caption{\texttt{ic3-load-$F_1$\labelalg{ic3_retiming_invar}}}
\end{algorithm}

By establishing that $I \implies \Inv$ (line \ref{line:i_implies_R}) and $I(S) \wedge T(S,X,S') \implies \Inv(S')$ (line \ref{line:i_and_T_implies_R})
before we let $F_1 := \Inv$ (line \ref{line:F1_equals_R}), we guarantee that $F_0 \implies F_1$, $F_0(S) \wedge T(S,X,S') \implies F_1(S')$.

We analyze
how \refalg{ic3_retiming_invar} behaves in the two relevant scenarios:
1) \refalg{retiming_invar} reports \emph{SUCCESS}, i.e. for the given $\Inv$ Conditions (\ref{inv1}) - (\ref{inv3})
hold 
and
2) \refalg{retiming_invar} reports \emph{FAIL} if $\Inv$ was not correctly retrieved.
\textbf{1)} Assume that \refalg{retiming_invar} reports \emph{SUCCESS}.
	We run \refalg{ic3_retiming_invar} on $\Inv$ and the same sequential miter.
	From Condition (\ref{inv1}) we know that $I \implies \Inv$ and from $I \implies \Inv$ and $\Inv(S) \wedge T(S,X,S') \implies \Inv(S')$ 
	(Condition (\ref{inv3})) we know that $ I(S) \wedge T(S,X, S') \implies \Inv(S')$. Thus, the checks in lines \ref{line:i_implies_R} and
	\ref{line:i_and_T_implies_R} are unsatisfiable and no conjunct will be removed from $\Inv$.
	We know from Condition (\ref{inv2}) that $\Inv \implies \cequiv$ and therefore, since $F_1 := \Inv$, the check
	$F_1(S) \wedge \neg \cequiv(S,X)$ (\refsec{ic3pdr}) will be unsatisfiable and a new frame $F_2 = \top$ is opened.
	IC3 will propagate \emph{all} (see \refsec{propagate}) clauses to $F_2$ because $F_1(S) \wedge T(S,X,S') \implies F_1(S')$
	follows from $\Inv(S) \wedge T(S,X,S') \implies \Inv(S')$ (Condition (\ref{inv3})).
	IC3 will terminate after \emph{one} propagation phase and conclude sequential equivalence.

    \textbf{2)} Assume that \refalg{retiming_invar} reports \emph{FAIL}. All conjuncts $r$ of $\Inv$ that violate $I \implies r$ or 
	$I(S) \wedge T(S, S') \implies r(S')$ are immediately sorted out in lines \ref{line:i_implies_R} and \ref{line:i_and_T_implies_R}.
    If a non-invariant conjunct $r$ excludes states which are reachable within up to $i$ steps, then 
    it will not be propagated any further than $F_{i-1}$. Hence, like any other non-invariant lemma learned by IC3 automatically, it does
	not affect soundness or completeness.
    If an $r$ is invariant, but $\neg r$ is reachable from $INV$ 
    (i.e., $INV$ is not yet strong enough to prove that $r$ is invariant), 
    then IC3 may derive
    lemmas $L$ such that $\neg r$ cannot be reached from  $INV \wedge L$.
	IC3 will strengthen ``the actually invariant 
    parts'' of $\Inv$ such that an inductive invariant results.

\texttt{ic3-load-}$F_1$ is a complete algorithm to prove or refute sequential equivalence. 
If $\Inv$ is an inductive invariant, it does not require more SAT calls than \refalg{retiming_invar}.

We do not provide a formal proof for correctness and termination, but the intuition is simple: we  only add clauses to $F_1$ that could be added by IC3 automatically as well. We guarantee
that the remaining conjuncts from $\Inv$ that we add represent an over-approximation of all reachable states in up to one step.

\paragraph{\textbf{Equivalence Classes}}
If we use sequential simulation as a means to guess equivalent signals, it will provide us
with a number of suspected equivalence classes. 
If there is an invariant consisting only of signal correspondences between FFs
and internal signals, including those added by virtual forward retiming
(this holds, if only retiming has been applied, but it may not be the case,
if logic synthesis has been applied as well), then
all relevant equivalences are equivalences between pairs of signals in those
equivalence classes, since the speculated equivalences overapproximate the real ones.
Thus, if \texttt{ic3-load-}$F_1$ then produces for each signal pair $(x,y)$ of each equivalence class
the two clauses $(\neg x \vee y)$ and $(x \wedge \neg y)$,\footnote{
Note that, technically speaking, the CNF representation of the transition relation $T$
includes Tseitin transformations of the circuit functions $\delta_x$ and $\delta_y$ computing 
signals $x$ and $y$,
such that $x \equiv y \wedge T$ establishes an equivalence $\delta_x \equiv \delta_y$.
This observation also holds for signals produced by virtual forward retiming.
}
then it will discard non-invariant correspondences 
either initially in  lines \ref{line:i_implies_R} / \ref{line:i_and_T_implies_R} of \refalg{ic3_retiming_invar}
or it will discard them by not further propagating them within IC3.

The number of signal pairs, however, is quadratic in the sizes of equivalence classes,
which may result in a large number of expensive calls to the SAT solver.
A solution could be to exploit transitivity by, e.g., representing a speculated equivalence 
class only with a linear number of equivalences with one representative of the equivalence class
(e.g. the earliest node wrt. a topological order of the sequential miter circuit).
However, if this representative does not occur in the subset of really equivalent signals for this equivalence class,  
then all clauses for this equivalence class would be discarded.

Simple \texttt{ic3-load-}$F_1$ 
showed weaknesses in first experimental evaluations, since it
either suffers from huge numbers of clauses inserted into $F_1$ or
essential equivalences may get lost when exploiting transitivity as described above. 
Therefore we come up with a refinement of the basic idea that mitigates those problems.

\paragraph{\textbf{\texttt{ic3-sim}}}

Our actual implementation \texttt{ic3-sim}, which is presented in \refalg{ic3_spec_invar}, is able to make use of the simulator 
for two purposes: Guessing initial equivalence classes and performing a targeted refinement
of the equivalence classes.

We introduce a procedure \emph{Refine}($s, \vec{x}, \Inv$) that takes a start state $s$, a (possibly empty) 
input sequence $\vec{x}$ and a suspected invariant $\Inv$ as input. It refines the equivalence classes that 
are represented by $\Inv$ by simulating from $s$ via $\vec{x}$ and rebuilds the clausal representation 
$\Inv$, if simulation finds new \emph{non-equivalences}.

We represent the guessed equivalence classes as aforementioned: We build equivalence
clauses between each element and the representative of the class.\footnote{ 
Note that our implementation only considers equivalences between FF outputs $s$ and signals $t$ depending
on FF outputs. Thus $\Inv$ only depends on $S$.}
In contrast to \texttt{ic3-load-}$F_1$, this CNF $\Inv(S)$ 
is not used to initialize $F_1$, but it is
used to restrict the behaviour of $T$ by letting 
$T^{\Inv}(S,X,S') = T(S,X,S') \wedge \Inv(S)$.

We retrieve the initial version of $\Inv$ by \textit{InitINV()} (line \ref{line:preproc}) that applies 
sequential simulation, possibly augmented by virtual forward retiming (\refsec{virt_retiming}).

Like in \texttt{ic3-load-}$F_1$, we remove each candidate $r_i(S)$  
where $I(S) \implies r_i(S)$ or $I(S) \wedge T(S,X,S') \implies r_i(S')$
does not hold.

We then let IC3 run as usual (line \ref{line:ic3_std}).
If IC3 terminates under $T^{\Inv}$ and reports `\emph{UNSAFE}' with a counterexample (line 
\ref{line:inequivalent}), then it is a real counterexample because $T^{\Inv}$ under-approximates the behaviour of $T$.

If it reports `\emph{SAFE}' in some frame $F_k$ though, then we only know that $F_k(S) \wedge\Inv(S) \wedge T(S,X,S') \implies F_k(S')$ but it is unknown whether $F_k(S) \wedge\Inv(S) \wedge T(S,X,S') \implies F_k(S')\wedge\Inv(S')$. Therefore, we have to check for each 
$r_i \in \Inv$ whether SAT?[$F_k(S) \wedge T^{\Inv}(S,X,S') \wedge \neg r_i(S')$].

If the call (line \ref{line:maypo_sat}) is \emph{satisfiable}, then refinement may become necessary.
We  treat $\neg r_i$ exactly like an IC3 POB (spawning and resolving POBs recursively
by calling \emph{ResolveObligations()}, see \refsec{resolve_obl}),
 with the only difference that upon detecting a \emph{counterexample}, i.e. a path from the 
 initial states $I$ to $\neg r_i$, we use that sequence of inputs to seed our simulator for refinement. 
 We call these POBs `may-POBs' as coined by \cite{DBLP:conf/fmcad/GurfinkelI15}.

Refinement is \emph{not necessarily} needed. The may-POB can also be proven unreachable
from the initial states, then \emph{ResolveObligations()} will return an empty counterexample state. 
Along the way of proving that, IC3 will have learned additional strengthening clauses. For efficiency
reasons, our implementation augments the original \emph{ResolveObligations()} by an upper bound of 
recursions. 
\emph{ResolveObligations}($(F_0, .., F_n),\neg r_i, k+1, \text{cutoff\_budget} = c$)
will try to resolve the POB $\neg r_i(S)$
only until $c$ new POBs are spawned recursively.
If the budget is exceeded, it will return the latest POB, and thus the POB in the earliest time frame that was considered. 
From this state (which is not an initial state) we reconstruct a path to  $\neg r_i(S)$ as a counterexample.
(In the special case that 
the cutoff budget is 0, \emph{ResolveObligations()} 
will return the counterexample to induction from the check 
SAT?[$F_k(S) \wedge \Inv(S) \wedge T(S,X,S') \wedge \neg r_i(S')$]
similar to \cite{DBLP:conf/dac/MonyBPK05,DBLP:conf/date/MonyBMB09}.)

Given that $r_i = (x \equiv y)$, 
refinement (line \ref{line:refine}) will in all cases \emph{at least} remove either $x$ or $y$ from its suspected equivalence class.

If at least one refinement step was triggered during iterating over $\Inv$, then
$\Inv$ is reduced and therefore weakened. Thus, learned clauses by IC3 under the old $\Inv$ can
become spurious. The best solution we found was also the simplest: We restart IC3 and try to salvage
as many learned clauses as possible (further explained in \refsec{ic3sim_restart}).

\begin{algorithm}[tb]
\SetInd{0.3em}{0.5em}
\begin{footnotesize}
$\Inv = \mathit{InitINV()}$\;\label{line:preproc}
\ForEach{conjunct $r_i$ in $\Inv$}{
	\lIf{($\text{init},-$) := SAT? [ $I(S) \wedge \neg r_i(S)$ ]}{
		\textit{Refine}($\text{init}, -, \Inv$)
	}
}
\ForEach{conjunct $r_i$ in $\Inv$}{
	\lIf{($\text{init}$,$x_0$) := SAT? [ $I(S) \wedge T(S,X, S') \wedge \neg r_i(S')$ ]}{
		\textit{Refine}($\text{init}, x_0, \Inv$)
	}
}
\While{\emph{true}}{
    $T^{\Inv} = T \wedge \Inv$; $F_0 = I(S)$; $F_1 = \top$; $n := 1$\;
	\eIf{IC3($I,\bm{T^{\Inv}},\cequiv, (F_0, \ldots, F_n)$) == \emph{`SAFE'}}{\label{line:ic3_std}
		spurious $:=$ \textit{false}\;
		$(F_0, \ldots, F_n) := $ frames of current IC3 execution\;
		\While{\textit{Propagate}($(F_0, \ldots, F_n), T^{\Inv}$) == \emph{`SAFE'}}{
		    $F_k :=$ possibly spurious ind. invariant in frame $k$\;
			possibly\_spurious $:=$ \textit{false}\;
			\ForEach{conjunct $r_i$ in $\Inv$}{
				\If{SAT? [ $F_k(S) \wedge T^{\Inv}(S,X,S') \wedge \neg r_i(S')$ ]}{
					\label{line:maypo_sat}
					possibly\_spurious $:= \mathit{true}$\;
					$(s, <x_1,\ldots, x_n>):=$ \textit{ResolveObligations}($(F_0, .., F_n),\neg r_i, k+1, \text{cutoff}$)\;
					\uIf{$\neg$ s.empty()}{
						\textit{Refine($s, <x_1,\ldots, x_n>, \Inv$)}; spurious := \textit{true}\;\label{line:refine}
					}
				}
			}
			\lIf{spurious}{ $F_1 = \top$; $n:=1$; \textbf{break} }\label{line:restart}
			\lIf{$\neg$possibly\_spurious}{	\Return \emph{`EQUIVALENT'} }\label{line:equivalent}
		}
	}{
		\Return \emph{`INEQUIVALENT'}\;\label{line:inequivalent}
	}
}
\end{footnotesize}
\caption{\texttt{ic3-sim}\labelalg{ic3_spec_invar}}
\end{algorithm}

Note that our simulations for refinement do not necessarily start with an initial state,
if the cutoff budget is exceeded, before we have found a path from an initial state to $\neg r_i$.
In that case, simulation may add \emph{``non-equivalences''}
which are not really proven to be non-equivalences. 
However, this does not affect the completeness of the 
method, since IC3 is complete anyway, even without (speculated) equivalences. 

Preliminary experimental results
showed that it pays off to make the removal of speculated equivalences more aggressive than 
needed, since it may be counterproductive to invest too much effort into the 
proof of equivalences that are too hard to prove and possibly not needed.

The higher the cutoff, the more effort we put into additional strengthening. In some
cases, like simple retiming verification, where aggressive refinement is helpful, this
might hurt efficiency. On the other hand, we experienced that it pays off if additional
stronger synthesis steps were applied to \revised{}.

The approach is \emph{complete}, because in the worst case $\Inv$ is weakened to the point that $T^{\Inv} = T$.
Therefore, standard IC3 would be applied eventually.

\paragraph{\textbf{Restart}}\labelsec{ic3sim_restart}
Restarting IC3 might hurt, if many additional lemmas were deduced by IC3 along the 
way. 
Therefore, before restart we check for each previously learnt lemma (with the same check as in lines \ref{line:i_implies_R} and \ref{line:i_and_T_implies_R} of \refalg{ic3_retiming_invar}) whether it is sound to 
add it to frame $F_1$. Thus, the new IC3 run does not start from scratch, but with the previously learnt lemmas which can be used in frame $F_1$. Later on, those lemmas may be propagated by 
\emph{Propagate()}.

\subsection{Retiming to Preprocess Sequential Equivalence Checking}\labelsec{retiming_preproc}

We have to prove that the output of the miter $M^{orig} := M(C^{\text{golden}}, C^{\text{revised}})$
remains at logic 0 under all possible input sequences.
As already mentioned in Sect. III.B., to reach this goal, sequential equivalence checkers like \texttt{dprove}
from ABC~\cite{DBLP:conf/cav/BraytonM10} often use retiming on
$M^{orig}$, resulting in a miter $M^{pp}$, as a preprocessing
step to make the equivalence checking problem easier.
\texttt{dprove} uses ``maximal forward retiming'', retiming minimizing the number
of FFs, and combinations thereof.
However, from a certification perspective, it is problematic to make use of
(uncertified) retiming to prove the correctness of (sequential logic synthesis
and) retiming.
In this paper we follow a different approach:
We use retiming as a preprocessing step as well, but we
produce an inductive invariant $INV^{pp}$ certifying the correctness
of $M^{pp}$ (as described in the previous sections)
\emph{and} we transform $INV^{pp}$ into an
inductive invariant $INV^{orig}$ certifying the correctness of $M^{orig}$.
In this section we show how to do this transformation and we prove the correctness
of the transformation.

Our approach is inspired by \cite{mneimneh2003reverse} where a retiming invariant
is computed for \golden{} and \revised{} under the assumption that 
\revised{} results from retiming only, see also Sect. III.A..
\cite{mneimneh2003reverse} gives an intuition for a proof, but neither the retiming invariant they present
is complete nor the proof is complete. 

Since the retiming for preprocessing is done within the verification tool, 
we have the advantage that we \emph{know} the series of retiming moves performed
by the verification tool.

We assume that we performed $n$ FF movements starting from 
the miter $M^{orig}$, finally arriving at the miter $M^{pp}$.
In addition to the FF movements, corresponding changes of 
the  initial states were performed.
We assume that we computed an inductive invariant $INV^{pp}$ for
$M^{pp}$. To transform $INV^{pp}$ into $INV^{orig}$, we consider
the \emph{reverse} series of $n$ FF movements turning $M^{pp}$ into
$M^{orig}$, i.e., we consider a series of miter circuits $M^{(0)}, \ldots,
M^{(n)}$ with $M^{(0)} := M^{pp}$, $M^{(n)} := M^{orig}$, and $M^{(i+1)}$
results from $M^{(i)}$ by a reverted FF movement.
The assignments to the initial states are reverted as well.
The corresponding equivalence predicates of the miters
$M^{(i)}$ are $\cequiv^{(i)}$ ($0 \leq i \leq n$).
We define $R^{(0)} := INV^{pp}$ and we construct (step by step)
$R^{(1)}, \ldots, R^{(n)}$ with $INV^{orig} := R^{(n)}$.

\begin{figure}[t]
\centering
\includegraphics[width=0.8\columnwidth]{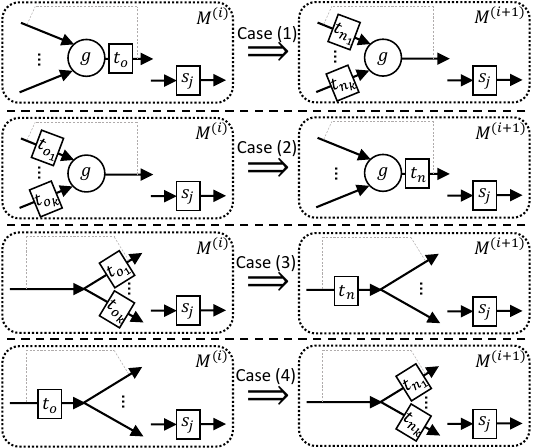}
\caption{Cases (1) -- (4) of retiming operations.}
\label{fig:cases}
\vspace*{-5mm}
\end{figure}

For the definition of $R^{(i+1)}$ from $R^{(i)}$ 
($0 \leq i \leq n-1$) and for the changed initial states we consider 4 cases, see
Fig.~\ref{fig:cases} for an illustration:

    (1) $M^{(i+1)}$ results from $M^{(i)}$ by backward retiming over some gate $g$ with $k$ inputs.
    Let $S_i = \{s_1, \ldots, s_{m_i}\} \dot\cup \{t_o\}$ be the states
    in $M^{(i)}$,   
    $t_o$ is the moved FF.
    Then $S_{i+1} = \{s_1, \ldots, s_{m_i}\} \dot\cup \{t_{n_1}, \ldots, t_{n_k}\}$ are the states in $M^{(i+1)}$  
    where $t_{n_1}, \ldots, t_{n_k}$ are the new FFs. 
    It holds 
    $i^{(i)}(s_j) = i^{(i+1)}(s_j)$ for all $1 \leq j \leq m_i$ and 
    $i^{(i)}(t_o) = g(i^{(i+1)}(t_{n_1}), \ldots, i^{(i+1)}(t_{n_k}))$, since the 
    original (forward) FF movement leading from $M^{(i+1)}$ to $M^{(i)}$ made exactly this 
    change to the initial state assignment.
    $R^{(i+1)} := \exists t_o [R^{(i)} \wedge (t_o \equiv g(t_{n_1}, \ldots,$ $t_{n_k}))]$.
    
    (2) $M^{(i+1)}$ results from $M^{(i)}$ by forward retiming over some gate $g$ with $k$ inputs.
    Let $S_i = \{s_1, \ldots, s_{m_i}\} \dot\cup \{t_{o_1}, \ldots,$ $t_{o_k}\}$ be the states in $M^{(i)}$, 
    $t_{o_1}, \ldots, t_{o_k}$ are the moved FFs. $t_{n}$ is the new FF.
    Then $S_{i+1} = \{s_1, \ldots, s_{m_i}\} \dot\cup \{t_{n}\}$
    are the states in $M^{(i+1)}$.
    It holds 
    $i^{(i)}(s_j) = i^{(i+1)}(s_j)$ for all $1 \leq j \leq m_i$ and 
    $i^{(i+1)}(t_n) = g(i^{(i)}(t_{o_1}), \ldots, i^{(i)}(t_{o_k}))$, since the 
    original (backward) FF movement leading from $M^{(i+1)}$ to $M^{(i)}$ made exactly this 
    change to the initial state assignment and it had been forbidden, if
    such an initial state assignment to $t_{o_1}, \ldots, t_{o_k}$ did not exist.
    $R^{(i+1)} := \exists t_{o_1}, \ldots, t_{o_k} [R^{(i)} \wedge (t_{n} \equiv g(t_{o_1}, \ldots, t_{o_k}))]$.
    
    (3) $M^{(i+1)}$ results from $M^{(i)}$ by backward retiming over some fanout with $k$ branches.
    Let $S_i = \{s_1, \ldots, s_{m_i}\} \dot\cup \{t_{o_1}, \ldots, t_{o_k}\}$  be the states in $M^{(i)}$, 
    $t_{o_1}, \ldots,$ $t_{o_k}$ are the moved FFs.
    $t_{n}$ is the new FF.
    Then $S_{i+1} = \{s_1, \ldots, s_{m_i}\} \dot\cup$ $\{t_{n}\}$
    are the states in $M^{(i+1)}$.
    It holds 
    $i^{(i)}(s_j) = i^{(i+1)}(s_j)$ for all $1 \leq j \leq m_i$ and 
    $i^{(i+1)}(t_n) = i^{(i)}(t_{o_1}) = \ldots = i^{(i)}(t_{o_k})$, since the 
    original (forward) FF movement leading from $M^{(i+1)}$ to $M^{(i)}$ made exactly this 
    change to the initial state assignment. 
    $R^{(i+1)} := \exists t_{o_1}, \ldots, t_{o_k} [R^{(i)} \wedge \bigwedge_{l=1}^k(t_{o_l} \equiv t_{n})]$.

    (4) $M^{(i+1)}$ results from $M^{(i)}$ by forward retiming over some fanout with $k$ branches.
    Let $S_i = \{s_1, \ldots, s_{m_i}\} \dot\cup \{t_{o}\}$
    be the states in $M^{(i)}$,   
    $t_{o}$ is the moved FF.
    $t_{n_1}, \ldots, t_{n_k}$ are the new FFs.
    Then $S_{i+1} = \{s_1, \ldots,$ $s_{m_i}\} \dot\cup \{t_{n_1}, \ldots, t_{n_k}\}$ are the states in $M^{(i+1)}$.
    It holds 
    $i^{(i)}(s_j) = i^{(i+1)}(s_j)$ for all $1 \leq j \leq m_i$ and 
    $i^{(i)}(t_o) = i^{(i+1)}(t_{n_1}) = \ldots = i^{(i+1)}(t_{n_k})$, since the 
    original (backward) FF movement leading from $M^{(i+1)}$ to $M^{(i)}$ made exactly this 
    change to the initial state assignment and the FF movement had been forbidden, if
    not all $i^{(i+1)}(t_{n_1}), \ldots, i^{(i+1)}(t_{n_k})$ would be equal.    
    $R^{(i+1)} := \exists t_{o} [R^{(i)} \wedge \bigwedge_{l=1}^k(t_{o} \equiv t_{n_l})]$.

\begin{theorem}
\labeltheorem{ret_inv}
If $R^{(0)}$ is an inductive invariant of the miter $M^{(0)}$
proving $\cequiv^{(0)}$, then
(a) $R^{(i)}$ holds in all initial states of $M^{(i)}$,
(b) $R^{(i)}$ implies $\cequiv^{(i)}$ and
(c) $R^{(i)}$ is inductive in the miter $M^{(i)}$
for $0 \leq i \leq n$.
\end{theorem}

The proof differentiates between four cases (1) to (4) of FF movements. For the definitions
of $R^{(i+1)}$ from $R^{(i)}$ in cases (1) to (4) 
and for the corresponding transformations of the initial states 
see \refsec{retiming_preproc}.
Fig.~\ref{fig:appendix} illustrates the FF movements in the different cases.

\begin{proof} 
We prove \reftheorem{ret_inv} by induction over $i$.
The base case follows easily from the fact that 
$R^{(0)} = INV^{pp}$ is an inductive invariant of $M^{(0)}$, i.e.,
it satisfies parts (a), (b), and (c) of \reftheorem{ret_inv}.

We start with the induction step for part (a).

Case~(1): We consider the initial state assignment 
$(i^{(i+1)}(s_1), \ldots, i^{(i+1)}(s_{m_1}), i^{(i+1)}(t_{n_1}), \ldots, i^{(i+1)}(t_{n_k}))$ 
of $M^{(i+1)}$.
We know from the induction assumption that 
the initial state
$(i^{(i)}(s_1), \ldots, i^{(i)}(s_{m_i}), i^{(i)}(t_o))$ of $M^{(i)}$
satisfies $R^{(i)}$.
Since $i^{(i)}(t_o) = g(i^{(i+1)}(t_{n_1}), \ldots, i^{(i+1)}(t_{n_k}))$ and
$i^{(i)}(s_j) = i^{(i+1)}(s_j)$ for all $1 \leq j \leq m_i$,
$(i^{(i+1)}(s_1),$ \linebreak
$\ldots, i^{(i+1)}(s_{m_i}), i^{(i)}(t_o), i^{(i+1)}(t_{n_1}), \ldots, i^{(i+1)}(t_{n_k}))$
satisfies
$R^{(i)} \wedge (t_o \equiv g(t_{n_1}, \ldots,$ $t_{n_k}))$.
Thus, 
$(i^{(i+1)}(s_1), \ldots, i^{(i+1)}(s_{m_i}), i^{(i+1)}(t_{n_1}), \ldots, i^{(i+1)}(t_{n_k}))$
satisfies
$\exists t_o [R^{(i)} \wedge (t_o \equiv g(t_{n_1}, \ldots,$ $t_{n_k}))] = R^{(i+1)}$.

\begin{figure*}[t]
\centering
\includegraphics[width=\textwidth]{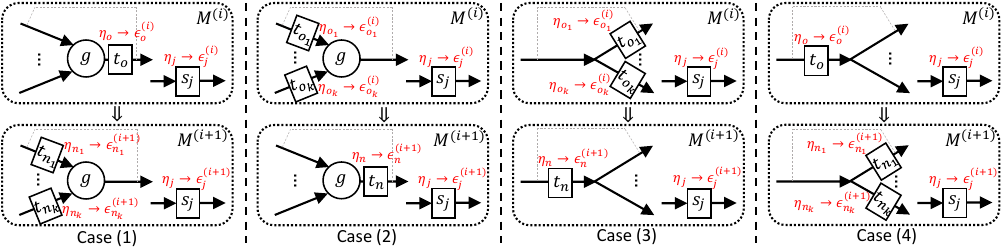}
\caption{Cases (1) -- (4) of retiming operations. The values for the FFs shown in red 
are chosen in the induction step for part (c) of \reftheorem{ret_inv}.}
\label{fig:appendix}
\end{figure*}

The proof for Cases (2), (3), and (4) is very similar:

Case~(2): By induction assumption the initial state 
$(i^{(i)}(s_1), \ldots, i^{(i)}(s_{m_i}), i^{(i)}(t_{o_1}), \ldots, i^{(i)}(t_{o_k}))$ 
of $M^{(i)}$ satisfies $R^{(i)}$.
Since $i^{(i+1)}(t_n) = g(i^{(i)}(t_{o_1}), \ldots, i^{(i)}(t_{o_k}))$ and 
$i^{(i)}(s_j) = i^{(i+1)}(s_j)$ for all $1 \leq j \leq m_i$,
$(i^{(i+1)}(s_1), \ldots, i^{(i+1)}(s_{m_i}), i^{(i)}(t_{o_1}), \ldots, i^{(i)}(t_{o_k}), i^{(i+1)}(t_n))$
satisfies 
$R^{(i)} \wedge (t_{n} \equiv g(t_{o_1}, \ldots, t_{o_k}))$.
Thus, the initial state 
$(i^{(i+1)}(s_1), \ldots, i^{(i+1)}(s_{m_i}), i^{(i+1)}(t_n))$
of $M^{(i+1)}$ satisfies
$\exists t_{o_1}, \ldots, t_{o_k} [R^{(i)} \wedge (t_{n} \equiv g(t_{o_1}, \ldots, t_{o_k}))] = R^{(i+1)}$.

Case~(3): By induction assumption the initial state
$(i^{(i)}(s_1), \ldots, i^{(i)}(s_{m_i}), i^{(i)}(t_{o_1}), \ldots, i^{(i)}(t_{o_k}))$ 
of $M^{(i)}$ satisfies $R^{(i)}$.
Since $i^{(i+1)}(t_n) = i^{(i)}(t_{o_1}) = \ldots = i^{(i)}(t_{o_k})$ and
$i^{(i)}(s_j) = i^{(i+1)}(s_j)$ for all $1 \leq j \leq m_i$,
$(i^{(i+1)}(s_1), \ldots, i^{(i+1)}(s_{m_i}), i^{(i)}(t_{o_1}), \ldots, i^{(i)}(t_{o_k}), i^{(i+1)}(t_n))$
satisfies 
$R^{(i)} \wedge \bigwedge_{l=1}^k(t_{o_l} \equiv t_{n})$.
Thus, the initial state 
$(i^{(i+1)}(s_1), \ldots, i^{(i+1)}(s_{m_i}), i^{(i+1)}(t_n))$
of $M^{(i+1)}$ satisfies
$\exists t_{o_1}, \ldots, t_{o_k} [R^{(i)} \wedge \bigwedge_{l=1}^k(t_{o_l} \equiv t_{n})] = R^{(i+1)}$.

Case~(4): By induction assumption the initial state
$(i^{(i)}(s_1), \ldots, i^{(i)}(s_{m_i}), i^{(i)}(t_o))$ 
of $M^{(i)}$ satisfies $R^{(i)}$.
Since $i^{(i)}(t_o) = i^{(i+1)}(t_{n_1}) = \ldots = i^{(i+1)}(t_{n_k}))$ and
$i^{(i)}(s_j) = i^{(i+1)}(s_j)$ for all $1 \leq j \leq m_i$,
$(i^{(i+1)}(s_1),$ \linebreak $\ldots, i^{(i+1)}(s_{m_i}), i^{(i)}(t_o), i^{(i+1)}(t_{n_1}), \ldots, i^{(i+1)}(t_{n_k}))$
satisfies
$R^{(i)} \wedge \bigwedge_{l=1}^k(t_{o} \equiv t_{n_l})$.
Thus, the initial state 
$(i^{(i+1)}(s_1), \ldots, i^{(i+1)}(s_{m_i}), i^{(i+1)}(t_{n_1}), \ldots, i^{(i+1)}(t_{n_k}))$
of $M^{(i+1)}$ satisfies
$\exists t_{o} [R^{(i)} \wedge \bigwedge_{l=1}^k(t_{o} \equiv t_{n_l})] = R^{(i+1)}$.

Now we consider the induction step for part (b).

Case~(1): The output function $\cequiv^{(i+1)}$ of the miter $M^{(i+1)}$
results from $\cequiv^{(i)}$ by the function composition 
$\cequiv^{(i+1)} = \cequiv^{(i)}|_{t_o \leftarrow g(t_{n_1}, \ldots, t_{n_k})}$.
We consider an arbitrary assignment  
$A^{(i+1)} = (\eta_1, \ldots, \eta_{m_i}, \eta_{n_1}, \ldots, \eta_{n_k}, \xi_1, \ldots, \xi_m)$ 
to the FFs in $M^{(i+1)}$ and to the primary inputs that satisfies $R^{(i+1)}$.
It follows from the semantics of existential quantification that this assignment can be
extended by an assignment $\eta_o$ to the moved FF $t_o$ such that
$R^{(i)} \wedge (t_o \equiv g(t_{n_1}, \ldots, t_{n_k}))$ is satisfied
by the resulting 
$A^{(i)} = (\eta_1, \ldots, \eta_{m_i}, \eta_{n_1}, \ldots, \eta_{n_k}, \eta_o, \xi_1, \ldots, \xi_m)$.
$A^{(i)}$ defines a complete assignment to the FFs both for $M^{(i+1)}$ and $M^{(i)}$ and to the
primary inputs.
Since $A^{(i)}$ satisfies $R^{(i)}$, we can conclude from the induction hypothesis 
that $A^{(i)}$ satisfies $\cequiv^{(i)}$.
Since $A^{(i)}$ satisfies $t_o \equiv g(t_{n_1}, \ldots, t_{n_k})$,
$A^{(i)}$ also satisfies $\cequiv^{(i)}|_{t_o \leftarrow g(t_{n_1}, \ldots, t_{n_k})} = \cequiv^{(i+1)}$.
Since $\cequiv^{(i+1)}$ does not depend on $t_o$,
also $A^{(i+1)}$ satisfies $\cequiv^{(i+1)}$.

Case~(2): The output function $\cequiv^{(i)}$ of the miter $M^{(i)}$
results from $\cequiv^{(i+1)}$ by the function composition 
$\cequiv^{(i)} = \cequiv^{(i+1)}|_{t_n \leftarrow g(t_{o_1}, \ldots, t_{o_k})}$.
We consider an arbitrary assignment  
$A^{(i+1)} = (\eta_1, \ldots, \eta_{m_i}, \eta_{n}, \xi_1, \ldots, \xi_m)$ 
to the FFs in $M^{(i+1)}$ and to the primary inputs that satisfies $R^{(i+1)}$.
It follows from the semantics of existential quantification that this assignment can be
extended by an assignment $(\eta_{o_1}, \ldots, \eta_{o_k})$ to the moved FFs 
$t_{o_1}, \ldots, t_{o_k}$ such that
$R^{(i)} \wedge (t_{n} \equiv g(t_{o_1}, \ldots, t_{o_k}))$ is satisfied
by the resulting 
$A^{(i)} = (\eta_1, \ldots, \eta_{m_i}, \eta_{n}, \eta_{o_1}, \ldots, \eta_{o_k}, \xi_1, \ldots, \xi_m)$.
$A^{(i)}$ defines a complete assignment to the FFs both for $M^{(i+1)}$ and $M^{(i)}$ and to the
primary inputs.
Since $A^{(i)}$ satisfies $R^{(i)}$, we can conclude from the induction hypothesis 
that $A^{(i)}$ satisfies $\cequiv^{(i)} = \cequiv^{(i+1)}|_{t_n \leftarrow g(t_{o_1}, \ldots, t_{o_k})}$.
Since $A^{(i)}$ satisfies $t_{n} \equiv g(t_{o_1}, \ldots, t_{o_k})$,
$A^{(i)}$ also satisfies $\cequiv^{(i+1)}$.
Since $\cequiv^{(i+1)}$ does not depend on $t_{o_1}, \ldots, t_{o_k}$,
also $A^{(i+1)}$ satisfies $\cequiv^{(i+1)}$.

Case~(3): The output function $\cequiv^{(i+1)}$ of the miter $M^{(i+1)}$
results from $\cequiv^{(i)}$ by the function composition 
$\cequiv^{(i+1)} = \cequiv^{(i)}|_{t_{o_1} \leftarrow t_n, \ldots, t_{o_k} \leftarrow t_n}$.
We consider an arbitrary assignment  
$A^{(i+1)} = (\eta_1, \ldots, \eta_{m_i}, \eta_{n}, \xi_1, \ldots, \xi_m)$ 
to the FFs in $M^{(i+1)}$ and to the primary inputs that satisfies $R^{(i+1)}$.
It follows from the semantics of existential quantification that this assignment can be
extended by an assignment $(\eta_{o_1}, \ldots, \eta_{o_k})$ to the moved FFs $t_{o_1}, \ldots, t_{o_k}$ 
such that
$R^{(i)} \wedge \bigwedge_{l=1}^k(t_{o_l} \equiv t_{n})$
is satisfied
by the resulting 
$A^{(i)} = (\eta_1, \ldots, \eta_{m_i}, \eta_{n}, \eta_{o_1}, \ldots, \eta_{o_k}, \xi_1, \ldots, \xi_m)$.
$A^{(i)}$ defines a complete assignment to the FFs both for $M^{(i+1)}$ and $M^{(i)}$ and to the
primary inputs.
Since $A^{(i)}$ satisfies $R^{(i)}$, we can conclude from the induction hypothesis 
that $A^{(i)}$ satisfies $\cequiv^{(i)}$.
Since $A^{(i)}$ satisfies $\bigwedge_{l=1}^k(t_{o_l} \equiv t_{n})$,
we have $\eta_{o_1} = \eta_{n}, \ldots, \eta_{o_k} = \eta_{n}$ and thus
$A^{(i)}$ also satisfies
$\cequiv^{(i)}|_{t_{o_1} \leftarrow t_n, \ldots, t_{o_k} \leftarrow t_n} = \cequiv^{(i+1)}$.
Since $\cequiv^{(i+1)}$ does not depend on $t_{o_1}, \ldots, t_{o_k}$,
also $A^{(i+1)}$ satisfies $\cequiv^{(i+1)}$.

Case~(4): The output function $\cequiv^{(i)}$ of the miter $M^{(i)}$
results from $\cequiv^{(i+1)}$ by the function composition 
$\cequiv^{(i)} = \cequiv^{(i+1)}|_{t_{n_1} \leftarrow t_o, \ldots, t_{n_k} \leftarrow t_o}$.
We consider an arbitrary assignment  
$A^{(i+1)} = (\eta_1, \ldots, \eta_{m_i}, \eta_{n_1}, \ldots, \eta_{n_k}, \xi_1, \ldots, \xi_m)$ 
to the FFs and to the primary inputs in $M^{(i+1)}$ that satisfies $R^{(i+1)}$.
It follows from the semantics of existential quantification that this assignment can be
extended by an assignment $\eta_{o}$ to the moved FF $t_o$ such that
$R^{(i)} \wedge \bigwedge_{l=1}^k(t_{o} \equiv t_{n_l})$
is satisfied by the resulting
$A^{(i)} = (\eta_1, \ldots, \eta_{m_i}, \eta_{n_1}, \ldots, \eta_{n_k}, \eta_{o}, \xi_1, \ldots, \xi_m)$.
$A^{(i)}$ defines a complete assignment to the FFs both for $M^{(i+1)}$ and $M^{(i)}$ and to the
primary inputs.
Since $A^{(i)}$ satisfies $R^{(i)}$, we can conclude from the induction hypothesis 
that $A^{(i)}$ satisfies $\cequiv^{(i)}$.
Since $A^{(i)}$ satisfies $\bigwedge_{l=1}^k(t_{o} \equiv t_{n_l})$,
we have $\eta_{n_1} = \eta_{o}, \ldots, \eta_{n_k} = \eta_{o}$ and thus
$A^{(i)}$ also satisfies $\cequiv^{(i+1)}$.
Since $\cequiv^{(i+1)}$ does not depend on $t_o$,
also $A^{(i+1)}$ satisfies $\cequiv^{(i+1)}$.

Finally, we consider the induction step for part (c).

Case~(1): 
We consider an arbitrary assignment 
$A^{(i+1)} = (\eta_1, \ldots, \eta_{m_i}, \eta_{n_1}, \ldots, \eta_{n_k})$ 
to the FFs in $M^{(i+1)}$ that satisfies $R^{(i+1)}$.
It follows from the semantics of existential quantification that this assignment can be
extended by an assignment $\eta_o$ to the moved FF $t_o$ such that
$R^{(i)} \wedge (t_o \equiv g(t_{n_1}, \ldots, t_{n_k}))$ is satisfied. 
We extend the assignment further by an assignment $(\xi_1, \ldots, \xi_m)$
to the primary input variables resulting in
$A^{(i)} = (\eta_1, \ldots, \eta_{m_i}, \eta_{n_1}, \ldots, \eta_{n_k}, \eta_o, \xi_1, \ldots, \xi_m)$.
$A^{(i)}$ defines a complete assignment to the FFs both for $M^{(i+1)}$ and $M^{(i)}$ and to the
primary inputs.
Let the next state assignment corresponding to $A^{(i)}$ in $M^{(i)}$ be
$A'^{(i)} := (\epsilon^{(i)}_1, \ldots, \epsilon^{(i)}_{m_i}, \epsilon^{(i)}_o)$
and let
$A'^{(i+1)} = (\epsilon^{(i+1)}_1, \ldots, \epsilon^{(i+1)}_{m_i}, \epsilon^{(i+1)}_{n_1}, \ldots, \epsilon^{(i+1)}_{n_k})$ be
the next state assignment corresponding to $A^{(i)}$, i.e., corresponding to
$(\eta_1, \ldots, \eta_{m_i}, \eta_{n_1}, \ldots, \eta_{n_k}, \xi_1, \ldots, \xi_m)$, in $M^{(i+1)}$.
From the induction hypothesis we can conclude that $A'^{(i)}$ satisfies $R^{(i)}$. We have to prove that $A'^{(i+1)}$ satisfies $R^{(i+1)}$.

For the common FFs $s_j$ ($1 \leq i \leq m_i$) of $M^{(i)}$ and
$M^{(i+1)}$ the transition functions $\delta^{(i+1)}_{s_j}$ of $M^{(i+1)}$
result from the transition functions $\delta^{(i)}_{s_j}$ of $M^{(i)}$ by
$\delta^{(i+1)}_{s_j} = \delta^{(i)}_{s_j}|_{t_o \leftarrow g(t_{n_1}, \ldots, t_{n_k})}$.
Since the primary input assignments $(\xi_1, \ldots, \xi_m)$ and the present state
assignments $(\eta_1, \ldots, \eta_{m_i})$ to the non-moved FFs are the same in $M^{(i)}$ and $M^{(i+1)}$ 
and $\eta_o$ has been chosen such that $\eta_o = g(\eta_{n_1}, \ldots, \eta_{n_k})$,
we have $\epsilon^{(i)}_j = \epsilon^{(i+1)}_j$.

For $t_o$ and $t_{n_1}, \ldots, t_{n_k}$ we cannot just assume that
$\delta^{(i)}_{t_o} = g(\delta^{(i+1)}_{t_{n_1}}, \ldots, \delta^{(i+1)}_{t_{n_k}})$, since
there may be cycles containing the gate $g$, such that the next state of $t_o$ depends on the
current state of $t_o$. Our argument has to be more precise:
Because of $\eta_o = g(\eta_{n_1}, \ldots, \eta_{n_k})$, the signal assignments
at the output of $t_o$ in $M^{(i)}$ and at the output of $g$ in $M^{(i+1)}$
are the same in the current state. 
Because of this and because the primary input assignments $(\xi_1, \ldots, \xi_m)$ and the present state
assignments $(\eta_1, \ldots, \eta_{m_i})$ to the non-moved FFs are the same in $M^{(i)}$ and $M^{(i+1)}$,
the values computed by $M^{(i)}$
at input $j$ ($1 \leq j \leq k$) of gate $g$ and by $M^{(i+1)}$ at the input of $t_{n_j}$ are the same.
(This holds even in the presence of cycles containing $g$ -- remember that due to the backward retiming rules,
there is no fanout between the output of gate $g$ and 
$t_o$ in $M^{(i)}$ and there is no fan-out between $t_{n_j}$ and the
input $j$ of $g$ ($1 \leq j \leq k$) in $M^{(i+1)}$, see also Fig.~\ref{fig:appendix}, Case~(1).)
For the next state values of 
$t_o$ in $M^{(i)}$ and $t_{n_j}$ ($1 \leq j \leq k$) in $M^{(i+1)}$
this immediately results in
$\epsilon^{(i)}_o = g(\epsilon^{(i+1)}_{n_1}, \ldots, \epsilon^{(i+1)}_{n_k})$.

Altogether, 
$(\epsilon^{(i)}_1, \ldots, \epsilon^{(i)}_{m_i}, 
\epsilon^{(i+1)}_{n_1}, \ldots, \epsilon^{(i+1)}_{n_k}, \epsilon^{(i)}_o)
= (\epsilon^{(i+1)}_1,$ \linebreak 
$\ldots, \epsilon^{(i+1)}_{m_i}, 
\epsilon^{(i+1)}_{n_1}, \ldots, \epsilon^{(i+1)}_{n_k}, \epsilon^{(i)}_o)
$
satisfies  
both $R^{(i)}$ and 
$(t_o$ \linebreak
$\equiv g(t_{n_1}, \ldots, t_{n_k}))$, i.e.,
$(\epsilon^{(i+1)}_1, \ldots, \epsilon^{(i+1)}_{m_i}, 
\epsilon^{(i+1)}_{n_1}, \ldots, \epsilon^{(i+1)}_{n_k})$
satisfies $R^{(i+1)} = \exists t_o [R^{(i)} \wedge (t_o \equiv g(t_{n_1}, \ldots,$ $t_{n_k}))]$. 

Cases (2), (3), and (4) are proven in a similar way:

Case~(2): We consider an arbitrary assignment 
$A^{(i+1)} = (\eta_1, \ldots, \eta_{m_i}, \eta_{n})$ 
to the FFs in $M^{(i+1)}$ that satisfies $R^{(i+1)}$.
It follows from the semantics of existential quantification that this assignment can be
extended by an assignment $(\eta_{o_1}, \ldots, \eta_{o_k})$ to the moved FFs 
$t_{o_1}, \ldots, t_{o_k}$ such that
$R^{(i)} \wedge (t_{n} \equiv g(t_{o_1}, \ldots, t_{o_k}))$ is satisfied. 
We extend the assignment further by an assignment $(\xi_1, \ldots, \xi_m)$
to the primary input variables resulting in
$A^{(i)} = (\eta_1, \ldots, \eta_{m_i}, \eta_{n}, \eta_{o_1}, \ldots, \eta_{o_k}, \xi_1, \ldots, \xi_m)$.
$A^{(i)}$ defines a complete assignment to the FFs both for $M^{(i+1)}$ and $M^{(i)}$ and to the
primary inputs.
Let the next state assignment corresponding to $A^{(i)}$ in $M^{(i)}$ be
$A'^{(i)} := (\epsilon^{(i)}_1, \ldots, \epsilon^{(i)}_{m_i}, \epsilon^{(i)}_{o_1}, \ldots, \epsilon^{(i)}_{o_k})$
and let
$A'^{(i+1)} = (\epsilon^{(i+1)}_1, \ldots, \epsilon^{(i+1)}_{m_i}, \epsilon^{(i+1)}_{n})$ be
the next state assignment corresponding to $A^{(i)}$, i.e., corresponding to
$(\eta_1, \ldots, \eta_{m_i}, \eta_{n}, \xi_1, \ldots, \xi_m)$, in $M^{(i+1)}$.
From the induction hypothesis we can conclude that $A'^{(i)}$ satisfies $R^{(i)}$. We have to prove that $A'^{(i+1)}$ satisfies $R^{(i+1)}$.

For the common FFs $s_j$ ($1 \leq i \leq m_i$) of $M^{(i)}$ and
$M^{(i+1)}$ we have
$\delta^{(i)}_{s_j} = \delta^{(i+1)}_{s_j}|_{t_n \leftarrow g(t_{o_1}, \ldots, t_{o_k})}$.
Since the primary input assignments $(\xi_1, \ldots, \xi_m)$ and the present state
assignments $(\eta_1, \ldots, \eta_{m_i})$ to the non-moved FFs are the same in $M^{(i)}$ and $M^{(i+1)}$ 
and $\eta_{o_1}, \ldots, \eta_{o_k}$ have been chosen such that $\eta_n = g(\eta_{o_1}, \ldots, \eta_{o_k})$,
we have $\epsilon^{(i)}_j = \epsilon^{(i+1)}_j$.

Again, for $t_n$ and $t_{o_1}, \ldots, t_{o_k}$ we cannot just assume that
$\delta^{(i+1)}_{t_n} = g(\delta^{(i)}_{t_{o_1}}, \ldots, \delta^{(i)}_{t_{o_k}})$, since
there may be cycles containing the gate $g$, such that the next state of $t_n$ depends on the
current state of $t_n$. 
Because of $\eta_n = g(\eta_{o_1}, \ldots, \eta_{o_k})$, the signal assignments
at the output of $t_n$ in $M^{(i+1)}$ and at the output of $g$ in $M^{(i)}$
are the same in the current state. 
Because of this and because the primary input assignments $(\xi_1, \ldots, \xi_m)$ and the present state
assignments $(\eta_1, \ldots, \eta_{m_i})$ to the non-moved FFs are the same in $M^{(i)}$ and $M^{(i+1)}$,
the values computed by $M^{(i+1)}$
at input $j$ ($1 \leq j \leq k$) of gate $g$ and by $M^{(i)}$ at the input of $t_{o_j}$ are the same.
(This holds even in the presence of cycles containing $g$ -- remember that due to the backward retiming rules,
there is no fanout between the output of gate $g$ and 
$t_n$ in $M^{(i+1)}$ and there is no fan-out between $t_{o_j}$ and the
input $j$ of $g$ ($1 \leq j \leq k$) in $M^{(i)}$, see also Fig.~\ref{fig:appendix}, Case~(2).)
For the next state values of 
$t_n$ in $M^{(i+1)}$ and $t_{o_j}$ ($1 \leq j \leq k$) in $M^{(i)}$
this immediately results in
$\epsilon^{(i+1)}_n = g(\epsilon^{(i)}_{o_1}, \ldots, \epsilon^{(i)}_{o_k})$.

Altogether, 
$(\epsilon^{(i)}_1, \ldots, \epsilon^{(i)}_{m_i}, 
\epsilon^{(i+1)}_{n}, \epsilon^{(i)}_{o_1}, \ldots, \epsilon^{(i)}_{o_k})
= (\epsilon^{(i+1)}_1, \ldots, \epsilon^{(i+1)}_{m_i}, \epsilon^{(i+1)}_{n}, \epsilon^{(i)}_{o_1}, \ldots, \epsilon^{(i)}_{o_k})$
satisfies  
both $R^{(i)}$ and $(t_{n} \equiv g(t_{o_1}, \ldots, t_{o_k}))$,  i.e.,
$(\epsilon^{(i+1)}_1, \ldots, \epsilon^{(i+1)}_{m_i}, \epsilon^{(i+1)}_{n})$
satisfies 
$R^{(i+1)} = \exists t_{o_1}, \ldots, t_{o_k} [R^{(i)} \wedge (t_{n} \equiv g(t_{o_1}, \ldots, t_{o_k}))]$.

Case~(3): We consider an arbitrary assignment 
$A^{(i+1)} = (\eta_1, \ldots, \eta_{m_i}, \eta_{n})$ 
to the FFs in $M^{(i+1)}$ that satisfies $R^{(i+1)}$.
It follows from the semantics of existential quantification that this assignment can be
extended by an assignment $(\eta_{o_1}, \ldots, \eta_{o_k})$ to the moved FFs 
$t_{o_1}, \ldots, t_{o_k}$ such that
$R^{(i)} \wedge \bigwedge_{l=1}^k(t_{o_l} \equiv t_{n})$
is satisfied. 
We extend the assignment further by an assignment $(\xi_1, \ldots, \xi_m)$
to the primary input variables resulting in
$A^{(i)} = (\eta_1, \ldots, \eta_{m_i}, \eta_{n}, \eta_{o_1}, \ldots, \eta_{o_k}, \xi_1, \ldots, \xi_m)$.
$A^{(i)}$ defines a complete assignment to the FFs both for $M^{(i+1)}$ and $M^{(i)}$ and to the
primary inputs.
Let the next state assignment corresponding to $A^{(i)}$ in $M^{(i)}$ be
$A'^{(i)} := (\epsilon^{(i)}_1, \ldots, \epsilon^{(i)}_{m_i}, \epsilon^{(i)}_{o_1}, \ldots, \epsilon^{(i)}_{o_k})$
and let
$A'^{(i+1)} = (\epsilon^{(i+1)}_1, \ldots, \epsilon^{(i+1)}_{m_i}, \epsilon^{(i+1)}_{n})$ be
the next state assignment corresponding to $A^{(i)}$, i.e., corresponding to
$(\eta_1, \ldots, \eta_{m_i}, \eta_{n}, \xi_1, \ldots, \xi_m)$, in $M^{(i+1)}$.
From the induction hypothesis we can conclude that $A'^{(i)}$ satisfies $R^{(i)}$. We have to prove that $A'^{(i+1)}$ satisfies $R^{(i+1)}$.

For the common FFs $s_j$ ($1 \leq i \leq m_i$) of $M^{(i)}$ and
$M^{(i+1)}$ we have
$\delta^{(i+1)}_{s_j} = \delta^{(i)}_{s_j}|_{t_{o_1} \leftarrow t_n, \ldots, t_{o_k} \leftarrow t_n}$.
Since the primary input assignments $(\xi_1, \ldots, \xi_m)$ and the present state
assignments $(\eta_1, \ldots, \eta_{m_i})$ to the non-moved FFs are the same in $M^{(i)}$ and $M^{(i+1)}$ 
and $\eta_{o_1}, \ldots, \eta_{o_k}$ have been chosen with $\eta_{o_1} = \eta_{n}, \ldots, \eta_{o_k} = \eta_{n}$,
we have $\epsilon^{(i)}_j = \epsilon^{(i+1)}_j$.

We have
$\delta^{(i+1)}_{t_n} = \delta^{(i)}_{t_{o_j}}|_{t_{o_1} \leftarrow t_n, \ldots, t_{o_k} \leftarrow t_n)}$
for all $1 \leq j \leq k$ (this holds even if there are cycles in $M^{(i)}$ containing some
FFs $t_{o_1} \ldots, t_{o_k}$).
Because of $\eta_{o_1} = \eta_{n}, \ldots, \eta_{o_k} = \eta_{n}$ and because
the primary input assignments $(\xi_1, \ldots, \xi_m)$ and the present state
assignments $(\eta_1, \ldots, \eta_{m_i})$ to the non-moved FFs are the same in $M^{(i)}$ and $M^{(i+1)}$,
we have $\epsilon^{(i+1)}_n = \epsilon^{(i)}_{o_j}$ for all $1 \leq j \leq k$.

Altogether, 
$(\epsilon^{(i)}_1, \ldots, \epsilon^{(i)}_{m_i}, 
\epsilon^{(i+1)}_{n}, \epsilon^{(i)}_{o_1}, \ldots, \epsilon^{(i)}_{o_k})
= (\epsilon^{(i+1)}_1, \ldots, \epsilon^{(i+1)}_{m_i}, \epsilon^{(i+1)}_{n}, \epsilon^{(i)}_{o_1}, \ldots, \epsilon^{(i)}_{o_k})$
satisfies  
both $R^{(i)}$ and $\bigwedge_{l=1}^k(t_{o_l} \equiv t_{n})$, i.e.,
$(\epsilon^{(i+1)}_1, \ldots, \epsilon^{(i+1)}_{m_i}, \epsilon^{(i+1)}_{n})$
satisfies 
$R^{(i+1)} = \exists t_{o_1}, \ldots, t_{o_k} [R^{(i)} \wedge \bigwedge_{l=1}^k(t_{o_l} \equiv t_{n})]$.

Case~(4): We consider an arbitrary assignment 
$A^{(i+1)} = (\eta_1, \ldots, \eta_{m_i}, \eta_{n_1}, \ldots, \eta_{n_k})$ 
to the FFs in $M^{(i+1)}$ that satisfies $R^{(i+1)}$.
It follows from the semantics of existential quantification that this assignment can be
extended by an assignment $\eta_{o}$ to the moved FF
$t_{o}$ such that
$R^{(i)} \wedge \bigwedge_{l=1}^k(t_{o} \equiv t_{n_l})$
is satisfied. 
We extend the assignment further by an assignment $(\xi_1, \ldots, \xi_m)$
to the primary input variables resulting in
$A^{(i)} = (\eta_1, \ldots, \eta_{m_i}, \eta_{n_1}, \ldots, \eta_{n_k}, \eta_{o}, \xi_1, \ldots, \xi_m)$.
$A^{(i)}$ defines a complete assignment to the FFs both for $M^{(i+1)}$ and $M^{(i)}$ and to the
primary inputs.
Let the next state assignment corresponding to $A^{(i)}$ in $M^{(i)}$ be
$A'^{(i)} := (\epsilon^{(i)}_1, \ldots, \epsilon^{(i)}_{m_i}, \epsilon^{(i)}_{o})$
and let
$A'^{(i+1)} = (\epsilon^{(i+1)}_1, \ldots, \epsilon^{(i+1)}_{m_i}, \epsilon^{(i+1)}_{n_1}, \ldots, \epsilon^{(i+1)}_{n_k})$ be
the next state assignment corresponding to $A^{(i)}$, i.e., corresponding to
$(\eta_1, \ldots, \eta_{m_i}, \eta_{n_1}, \ldots, \eta_{n_k}, \xi_1, \ldots, \xi_m)$, in $M^{(i+1)}$.
From the induction hypothesis we can conclude that $A'^{(i)}$ satisfies $R^{(i)}$. We have to prove that $A'^{(i+1)}$ satisfies $R^{(i+1)}$.

For the common FFs $s_j$ ($1 \leq i \leq m_i$) of $M^{(i)}$ and
$M^{(i+1)}$ we have
$\delta^{(i)}_{s_j} = \delta^{(i+1)}_{s_j}|_{t_{n_1} \leftarrow t_o, \ldots, t_{n_k} \leftarrow t_o}$.
Since the primary input assignments $(\xi_1, \ldots, \xi_m)$ and the present state
assignments $(\eta_1, \ldots, \eta_{m_i})$ to the non-moved FFs are the same in $M^{(i)}$ and $M^{(i+1)}$ 
and $\eta_{o}$ has been chosen with $\eta_{o} = \eta_{n_1}, \ldots, \eta_{o} = \eta_{n_k}$,
we have $\epsilon^{(i)}_j = \epsilon^{(i+1)}_j$.

We have
$\delta^{(i)}_{t_o} = \delta^{(i+1)}_{t_{n_j}}|_{t_{n_1} \leftarrow t_o, \ldots, t_{n_k} \leftarrow t_o}$
for all $1 \leq j \leq k$ (this holds even if there are cycles in $M^{(i)}$ containing some
FF $t_{o}$).
Because of $\eta_{o} = \eta_{n_1}, \ldots, \eta_{o} = \eta_{n_k}$ and because
the primary input assignments $(\xi_1, \ldots, \xi_m)$ and the present state
assignments $(\eta_1, \ldots, \eta_{m_i})$ to the non-moved FFs are the same in $M^{(i)}$ and $M^{(i+1)}$,
we have $\epsilon^{(i)}_o = \epsilon^{(i+1)}_{n_j}$ for all $1 \leq j \leq k$.

Altogether, 
$(\epsilon^{(i)}_1, \ldots, \epsilon^{(i)}_{m_i}, 
\epsilon^{(i+1)}_{n_1}, \ldots, \epsilon^{(i+1)}_{n_k}, \epsilon^{(i)}_{o})
= (\epsilon^{(i+1)}_1,$ \linebreak$\ldots, \epsilon^{(i+1)}_{m_i}, \epsilon^{(i+1)}_{n_1}, \ldots, \epsilon^{(i+1)}_{n_k}, \epsilon^{(i)}_{o})$
satisfies  
both $R^{(i)}$ and $\bigwedge_{l=1}^k(t_{o_l} \equiv t_{n})$, i.e.,
$(\epsilon^{(i+1)}_1, \ldots, \epsilon^{(i+1)}_{m_i}, \epsilon^{(i+1)}_{n_1}, \ldots, \epsilon^{(i+1)}_{n_k})$
satisfies 
$R^{(i+1)} = \exists t_{o} [R^{(i)} \wedge \bigwedge_{l=1}^k(t_{o} \equiv t_{n_l})]$.
\end{proof} 

\reftheorem{ret_inv} implies in particular that $R^{(n)}$ is an inductive invariant of  
$M^{(n)}$, if $R^{(0)}$ is an inductive invariant of $M^{(0)}$.

Finally, we see that under certain conditions the existential quantifications within 
the retiming invariants can be omitted. This observation is based on the following lemma:
\begin{lemma}
\label{lemma:substitution}
If a predicate $R$ has the form $R = R' \wedge (v \equiv T)$ with
predicates $R'$ and $T$ where $T$ does not depend on $v$,
then $\exists v R$ and $R'|_{v \leftarrow T}$ are logically equivalent.
\end{lemma}

Lemma~\ref{lemma:substitution} immediately implies that all existential quantifications
for Cases (1), (3), and (4) of the definition of $R^{(i+1)}$ from $R^{(i)}$ above can be
replaced by substitutions. 
E.g. for backward retiming over a gate $g$ from Case (1)
$\exists t_o [R^{(i)} \wedge (t_o \equiv g(t_{n_1}, \ldots,$ $t_{n_k}))]$
can be replaced by
$R^{(i)}|_{t_o \leftarrow g(t_{n_1}, \ldots, t_{n_k})}$.
In Case (2), the question whether existential quantification can be
replaced by substitution depends on the concrete form of the predicate $R^{(i)}$.

If the preprocessing step uses forward retiming only, e.g., then the
reverse series of FF movements uses backward retiming only and thus it is
guaranteed that all existential quantifications can be replaced by substitutions.

Note that \reftheorem{ret_inv} is general enough that it can be applied to
retiming invariants in the sense of \cite{mneimneh2003reverse} as well.
We just have to define an original miter $M^{(0)}$ as a miter between $C^{\text{golden}}$
and a copy $C'^{\text{golden}}$. If $C^{\text{golden}}$ has FFs $s_1, \ldots, s_n$ and
its copy has FFs $s'_1, \ldots, s'_n$, then the original inductive invariant is
$R^{(0)} := \bigwedge_{i=1}^n (s_i \equiv s'_i)$. It is clear that
$R^{(0)}$ is an inductive invariant of $M^{(0)}$.
The series of FF movements is now chosen as a series of FF movements transforming 
the \emph{copy} of \golden{} into \revised{} . $M^{(n)}$  is then a miter between
\golden{} into \revised{} and, according to \reftheorem{ret_inv}, $R^{(n)}$ is 
an inductive invariant of $M^{(n)}$.

\subsection{Simple Path Constraints}\labelsec{simple_path}
We present a trivial example that showcases why simple path (also called unique state) constraints can become 
necessary in order to prove signal correspondence with $k$-induction.

\begin{figure}[tb]
    \pgfdeclarehorizontalshading{redgreen}{100bp}{
  color(0bp)=(red!20);
  color(50bp)=(red!20);
  color(50bp)=(green!20);
  color(100bp)=(green!20)
}
\pgfdeclarehorizontalshading{redblue}{100bp}{
  color(0bp)=(red!20);
  color(50bp)=(red!20);
  color(50bp)=(blue!20);
  color(100bp)=(blue!20)
}

\hfill\textbf{\golden{}}

\tikzmath{\scale = 1.3; \fdist = \scale ; \fonedist = \fdist * 1.8; \ftwodist = \fonedist * 2.1;\idist = \fdist / 2; \ffwidth = \scale * 4; \ffheight = \scale * 5.5;}
\begin{tikzpicture}[>=latex, thick]
  \node (i1) {};
  \node[below=\idist cm of i1] (i2) {$x_1$};
  \node[below=\idist cm of i2] (i3) {$x_2$};

  \node[and gate US, draw, logic gate inputs=nn, right=\fonedist cm of i1, yshift=-0.4cm, minimum size=8mm] (f1) {G};
  \draw (i2) -- ++(2,0) |- (f1.input 2);

  \node[and gate US, draw, logic gate inputs=nn, right=\ftwodist cm of i2, minimum size=8mm] (f2) {H};
  \draw (f1) -- node[pos=0.5, draw, rectangle, minimum width=\ffwidth mm, minimum height=\ffheight mm, fill=red!20]{{ $s_1$}}
                    node[pos=1, draw, circle, minimum size=1 mm, inner sep=0pt, fill=black](fanout){} ++(2,0) |- (f2.input 1);
  \draw (i3) -- ++(4,0) |- (f2.input 2);
  \node[not gate US, draw, logic gate inputs=n, minimum size=4mm, left=1cm of f1.input 1] (notgate) {};
  \draw (fanout) -- ++(0,0.8) -| (i1.east) -- ++(0,0) |- (notgate.input);
  \draw (notgate.output) -- ++(0.5,0)  |- (f1.input 1);

  \node[right=\fdist cm of f2] (o) {};
  \draw (f2.output) -- node[pos=0.5, draw, rectangle, minimum width=\ffwidth mm, minimum height=\ffheight mm, fill=blue!20]{{ $s_2$}} (o);

\end{tikzpicture}

\hfill\textbf{Intermediate Step}

\tikzmath{\scale = 1.3; \fdist = \scale ; \fonedist = \fdist * 1.8; \ftwodist = \fonedist * 2.1;\idist = \fdist / 2; \ffwidth = \scale * 4; \ffheight = \scale * 5.5;}
\begin{tikzpicture}[>=latex, thick, color=gray]
  \node (i1) {};
  \node[below=\idist cm of i1] (i2) {$x_1$};
  \node[below=\idist cm of i2] (i3) {$x_2$};

  \node[and gate US, draw, logic gate inputs=nn, right=\fonedist cm of i1, yshift=-0.4cm, minimum size=8mm] (f1) {G};
  \draw (i2) -- node[pos=0.3, draw, rectangle, minimum width=\ffwidth mm, minimum height=\ffheight mm, fill=red!20]{{ $\revisedff{s}_1$}} ++(2,0) |- (f1.input 2);

  \node[and gate US, draw, logic gate inputs=nn, right=\ftwodist cm of i2, minimum size=8mm] (f2) {H};
  \draw (f1.output) -- node[pos=1, draw, circle, minimum size=1 mm, inner sep=0pt, fill=gray](fanout){} ++(1.4,0) |- (f2.input 1);
  \draw (i3) -- ++(4,0) |- (f2.input 2);
  \node[not gate US, draw, logic gate inputs=n, minimum size=4mm, left=1cm of f1.input 1] (notgate) {};
  \draw (fanout) -- ++(0,0.8) -| node[pos=0.1, draw, rectangle, minimum width=\ffwidth mm, minimum height=\ffheight mm, fill=red!20]{{ $\revisedff{s}^{\star}_1$}} (i1.east) -- ++(0,0) |-  (notgate.input);
  \draw (notgate.output) -- ++(0.5,0)  |- (f1.input 1);

  \node[right=\fdist cm of f2] (o) {};
  \draw (f2.output) -- node[pos=0.5, draw, rectangle, minimum width=\ffwidth mm, minimum height=\ffheight mm, fill=blue!20]{{ $s_2$}} (o);

\end{tikzpicture}

\hfill\textbf{\revised{}}

\begin{tikzpicture}[>=latex, thick]
  \node (i1) {};
  \node[below=\idist cm of i1] (i2) {$x_1$};
  \node[below=\idist cm of i2] (i3) {$x_2$};

  \node[and gate US, draw, logic gate inputs=nn, right=\fonedist cm of i1, yshift=-0.4cm, minimum size=8mm] (f1) {G};
  \draw (i2) -- node[pos=0.3, draw, rectangle, minimum width=\ffwidth mm, minimum height=\ffheight mm, fill=red!20]{{ $\revisedff{s}_1$}}
                    node[pos=0.75, draw, rectangle, minimum width=\ffwidth mm, minimum height=\ffheight mm, shading=redblue]{{ $\revisedff{s}_4$}}  ++(2.2,0) |- (f1.input 2);

  \node[and gate US, draw, logic gate inputs=nn, right=\ftwodist cm of i2, minimum size=8mm] (f2) {H};
  \draw (f1.output) -- node[pos=1, draw, circle, minimum size=1 mm, inner sep=0pt, fill=black](fanout){} ++(1.4,0) |- (f2.input 1);
  \draw (i3) -- node[pos=0.2, draw, rectangle, minimum width=\ffwidth mm, minimum height=\ffheight mm, fill=blue!20]{{ $\revisedff{s}_2$}} ++(4,0) |- (f2.input 2);
  \node[not gate US, draw, logic gate inputs=n, minimum size=4mm, left=1cm of f1.input 1] (notgate) {};

  \draw (fanout) -- ++(0,0.8) -| node[pos=0.1, draw, rectangle, minimum width=\ffwidth mm, minimum height=\ffheight mm, shading=redblue]{{ $\revisedff{s}_3$}} (i1.east) -- ++(0,0) |-  (notgate.input);
  \draw (notgate.output) -- ++(0.5,0)  |- (f1.input 1);

  \node[right=\fdist cm of f2] (o) {};
  \draw (f2.output) -- (o);

\end{tikzpicture}
    \caption{A backward retimed circuit.\labelfig{simple_path}}
    \vspace*{-5mm}
\end{figure}
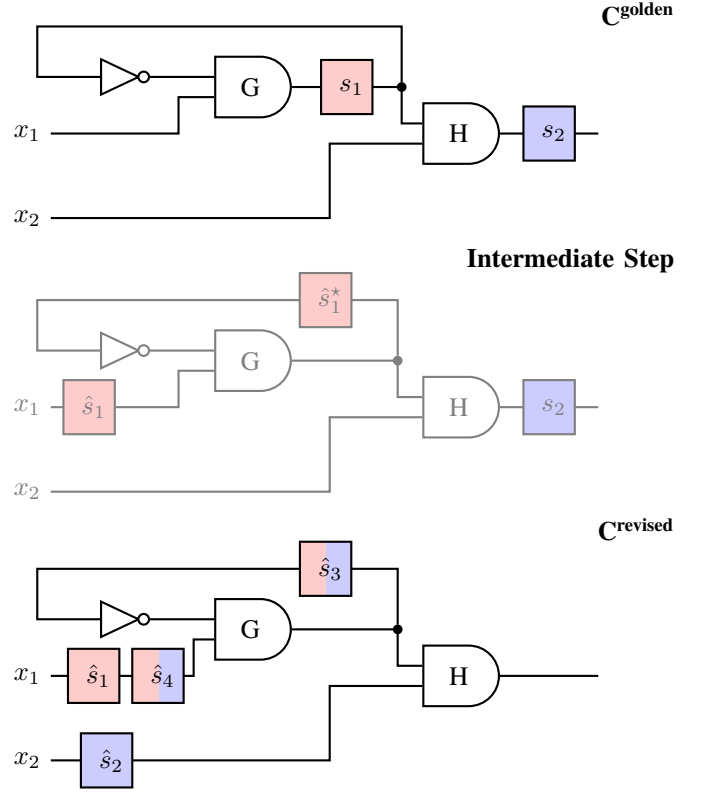

\begin{example}\labelexample{simple_path}
We consider the two \emph{sequentially equivalent} circuits in \reffig{simple_path}. The circuit \revised{} results from \golden{} by only retiming.
We retime $s_1$ backward over $G$ and the inverter. 
Thus, an intermediate flip-flop 
$\revisedff{s}^{\star}_1$ results.
Then we retime $s_2$ over $H$. As a result, we can retime $\revisedff{s}^{\star}_1$, together with the result of $s_2$ that was moved over $H$, over $G$ 
(and the inverter) again.
Then we arrive at \revised{}.
We assume an all-zero initial state and further assume that the correspondence $s_2 = \neg \revisedff{s}_3 \wedge \revisedff{s}_4 \wedge \revisedff{s}_2$ has to be proved. 

Now the interesting observation is the following:

There is an unreachable state $u = (s_1 = 0, s_2 = 0, \revisedff{s}_1 = 1, \revisedff{s}_2 = 0, \revisedff{s}_3 = 1, \revisedff{s}_4 = 1)$ that obviously satisfies the
equivalence of both outputs (both zero).

Starting from $u$ we arrive with input  $(x_1 = 1, x_2 = 0)$
at
$u' = (s_1 = 1, s_2 = 0, \revisedff{s}_1 = 1, \revisedff{s}_2 = 0, \revisedff{s}_3 = 0, \revisedff{s}_4 = 1)$
where both outputs remain zero.
From $u'$ we come back to $u$ with input $(x_1 = 1, x_2 = 0)$ again, i.e.,
there is an infinite (unreachable) loop of states with equivalent outputs.

If we start from $u$ and apply the inputs $x_1 = 0, x_2 = 1$ however, we arrive at the state
$(s_1 = 0, s_2 = 0, \revisedff{s}_1 = 0, \revisedff{s}_2 = 1, \revisedff{s}_3 = 0, \revisedff{s}_4 = 1)$ which violates equivalence on the outputs 
(\golden{} outputs 0 and \revised{} outputs 1).

Since there is no reachable state with non-equivalent outputs, the circuits are equivalent.
However, $k$-induction without simple path constraints will no terminate, since there is
an arbitrarily long chain of states with equivalent outputs which is followed by a state
with non-equivalent outputs.
Thus, simple path constraints are needed for $k$-induction. 
\end{example} 

\subsection{Certification}
Proofs of equivalence found by IC3 are inductive invariants.
Therefore, it is trivial to use a proof checker like \emph{Certifaiger} \cite{DBLP:conf/cav/FroleyksYPBH25} in the standard case of \texttt{ic3-sim}.
There is no need to certify multiple $k$-induction proofs or intermediate resynthesis steps.

In the \texttt{ic3-sim-F} case however, with maximal forward retiming, we can transform $\Inv^{pp}$
of the preprocessed miter $M^{pp}$ into an invariant $\Inv^{\text{orig}}$ of the original 
miter $M^{\text{orig}}$ (\refsec{retiming_preproc}).

In \refsec{certification} we demonstrate that all proofs could be checked with acceptable time / memory overhead. 
\section{Experimental Results}
\labelsec{results}

For our case study, we use a non-trivial collection of well known ISCAS'89 and ITC'99 \cite{DBLP:journals/dt/CornoRS00} benchmarks as well as 
open cores that were considered, e.g., at IWLS'05.
All synthesis operations, including retiming, are performed by ABC\footnote{commit 9478c172881f7c392493420f639a940fdfeb9a00} \cite{DBLP:conf/cav/BraytonM10}.
As the underlying SAT 
solver of \texttt{ic3-sim} we use 
\textsc{MiniSat} \cite{DBLP:conf/sat/EenS03}. 
Experiments with CaDiCaL \cite{DBLP:conf/cav/BiereFFFFP24} showed similar, but slightly longer run times.
All of our implementations and retimed / resynthesized AIGER \cite{Biere-FMV-TR-11-2} benchmarks can be found on github under \url{https://github.com/tobiasseufert/SEC-ic3-sim}. 

As a certificate format / checker, we use \emph{Certifaiger} \cite{DBLP:conf/cav/FroleyksYPBH25}.

We consider two of our
implementations: 
\texttt{ic3-sim} (\refalg{ic3_spec_invar}) detects the \emph{retiming depth} of the problem and performs
\emph{multi-timeframe} rarity simulation \cite{mishchenko2012using},
then it expands candidates for delayed equivalences via virtual forward retiming.
\texttt{ic3-sim-F} applies most forward retiming to the whole miter instead 
of lazy \emph{virtual} forward retiming. Therefore, only basic single time frame rarity
simulation is used.

For comparison, we consider the certifying general purpose model checker
rIC3 \cite{DBLP:conf/cav/SuYCBH25}. 

As a reference, we consider ABC's \emph{non-certifying} sequential 
equivalence checker \texttt{dprove}. 

As an additional reference, we present the results of the standard IC3 version that our 
implementations of \texttt{ic3-sim} and \texttt{ic3-sim-F} are based on. In the following
we denote our standard IC3 with \texttt{ic3}.

If not otherwise stated, we use a \emph{wall clock} timeout of 1\,h.
Memory limit is 4\,GB per thread. rIC3 spawns a total of 16 threads and is therefore allowed 
approximately 16 CPU-hours and 64\,GB of memory. 
Nevertheless, we report the actual CPU usage when indicating CPU times.

\paragraph*{\textbf{rIC3}}
is currently the strongest available \emph{certifying} hardware model checker.
It runs multiple IC3 variants concurrently. Some explicitly target \emph{equivalence checking}
problems by including internal signals in the state space \cite{DBLP:conf/fmcad/DurejaGIV21}.
Furthermore, it uses $k$-induction \cite{DBLP:conf/fmcad/SheeranSS00,DBLP:conf/fmcad/BjesseC00} but \emph{without}
simple-path constraints\footnote{Version 1.5.2, commit 17ef2d4356d2765cda432dbfbbd13773378c3e67, is currently not able to certify results under simple-path constraints.}. 

rIC3's $k$-induction engine runs Kissat \cite{kissat} with congruence closure \cite{DBLP:conf/sat/BiereFFF24}, which is well suited to
problems with many equivalent signals, such as sequential equivalence checking.

In total, rIC3 spawns 16 threads; our approach uses only one.

\paragraph*{\textbf{ABC's \texttt{dprove}}}
is a \emph{non-certifying} equivalence-checking engine.
It builds on the industrial-strength ABC tool \cite{DBLP:conf/cav/BraytonM10} and uses a whole
(single-threaded) portfolio of different techniques.
Its pre- and inprocessing uses a combination of maximal forward retiming,
minimal-area retiming, and logic-synthesis steps to simplify the transition relation,
which yields greatly simplified sequential equivalence checking problems. This is combined with induction queries
that incrementally unroll the transition relation (similar to $k$-induction without
unique-state constraints) to prove pairs of signals equivalent. Eventually, if the problem
remains unsolved, the simplified equivalence-checking problem is forwarded to interpolation-based model
checking or general IC3.

Certificates for \texttt{dprove}'s workflow are non-trivial, if not infeasible, with the current state of the art.

We nevertheless keep the comparison, even though \texttt{dprove} \emph{is not able to emit} certificates.

\subsection{Experiments}

\paragraph{\textbf{Retiming and Resynthesis Verification}}\labelsec{ret_res_std}
We start by comparing the \emph{whole portfolio} of rIC3 with 
\texttt{ic3-sim} on the sequential miter of the specification with \emph{only retimed}
and \emph{retimed and combinationally resynthesized} benchmarks, respectively. We 
apply the following abc commands:
\texttt{retime;dc2;rwsat;resyn}.

In the following, we give results for the version of our tool that we consider to be best suited
for the given problem set. For instance, plain retiming tasks are a good fit
for \texttt{ic3-sim-F}. Adding combinational resynthesis though may hurt maximal forward
retiming, therefore we use targeted virtual forward retiming, denoted simply by \texttt{ic3-sim}.
Also a \emph{cutoff = 0} (\refsec{ssweep_ic3}) is beneficial if we expect that a 
retiming invariant only with equivalences between FF output and signals (including signals from virtual forward retiming)
exists.

For retiming and resynthesis, we also consider \texttt{dprove} as a reference to showcase that not all
instances can be solved by the (to the best of our knowledge) best openly available sequential equivalence
checker. 

\begin{table*}[tb]
\begin{footnotesize}
\begin{center}
\setlength{\tabcolsep}{1pt}
\renewcommand{\arraystretch}{0.88}
\caption{Experimental results. CPU time in $s$. Each experiment is also performed on a \emph{pipelined} version of the circuits.\labeltab{results_merge}}
\begin{tabular}{l|rrr|rrrr|rrrr|rrrr}
  & \multicolumn{3}{c|}{\emph{only retiming}} & \multicolumn{4}{c|}{\emph{only retiming \textbf{pipelined}}} & \multicolumn{4}{c|}{\emph{retiming + comb. resyn.}} & \multicolumn{4}{c}{\parbox{32mm}{\centering\emph{retiming + comb. resyn.}\\\emph{\textbf{pipelined}}}} \\
\hline
\textbf{Benchmark} & \mcl{\textbf{ic3}} & \mcl{\textbf{rIC3}} & \mclv{\textbf{ic3-sim-F}}
  & \mcl{\textbf{ic3}} & \mcl{\textbf{rIC3}} & \mcl{\textbf{ic3-sim-F}} & \mclv{\parbox{14mm}{\centering\textbf{rIC3-kind}\\simple-path}}
  & \mcl{\textbf{ic3}} & \mcl{\textbf{rIC3}} & \mcl{\textbf{ic3-sim}} & \mclv{\parbox{12mm}{\centering{dprove} \\ \textbf{(uncert.)}}}
  & \mcl{\textbf{ic3}} & \mcl{\textbf{rIC3}} & \mcl{\textbf{ic3-sim}} & \mcl{\parbox{12mm}{\centering{dprove} \\ \textbf{(uncert.)}}} \\
\hline
aes\_core         & 901.99 & 5834.57    & 14.02   & --      & --           & 19.71         & -- & --      & 33684.15 & 18.39        & (4.93)      & --     & --       & 32.83    & (3.76)  \\
b12               & 798.17 & 5101.2     & 0.4     & --      & --           & 2             & -- & 1313.83 & 6318.04  & 1.12         & (10.35)     & --     & --       & 56.08    & (0.2)   \\
b13               & 78.46  & 97.18      & 0.13    & --      & 127.93       & 0.23          & -- & 25.22   & 163.64   & 0.2          & (52.32)     & --     & 120.96   & 1.22     & (--)    \\
b14               & 182.92 & 1299.12    & 1.21    & --      & 10326.3      & 5.31          & -- & 86.84   & 1480.49  & 4.56         & (1.03)      & --     & 8716.29  & 42.23    & (0.44)  \\
b15               & --     & --         & 4.6     & --      & --           & 15.3          & -- & --      & --       & 64.26        & (6.07)      & --     & --       & 299.27   & (3.89)  \\
b17               & --     & --         & 47.19   & --      & --           & 115.87        & -- & --      & --       & 1294.53      & (11.06)     & --     & --       & 2785.12  & (6.84)  \\
b18               & --     & 31911.97   & 1299.77 & --      & 49340.68     & 2698.73       & -- & --      & 44842.86 & --           & (18.63)     & --     & --       & --       & (22.25) \\
b20               & --     & 9875.31    & 2.77    & --      & --           & 12.8          & -- & --      & 11043.31 & 230.46       & (28.03)     & --     & --       & 29.71    & (1.32)  \\
b21               & --     & --         & 2.91    & --      & --           & 13.07         & -- & --      & --       & 131.73       & (30.44)     & --     & --       & 117.9    & (1.55)  \\
b22               & --     & 39234.32   & 6.35    & --      & --           & 20.3          & -- & --      & 27488.37 & 2228.81      & (58.86)     & --     & --       & 114.37   & (2.86)  \\
des\_perf         & --     & --         & 16.75   & --      & --           & 38.06         & -- & --      & --       & 94.76        & (32.62)     & --     & --       & --       & (4.28)  \\
mem\_ctrl         & --     & --         & 21.13   & --      & --           & 51.67         & -- & --      & --       & 63.23        & (8.83)      & --     & --       & --       & (--)    \\
pci\_bridge32     & --     & --         & 167.86  & --      & --           & 107           & -- & --      & --       & --           & (--)        & --     & --       & --       & (--)   \\
s15850            & --     & --         & 6.06    & --      & --           & 18.95         & -- & --      & --       & 52.43        & (0.68)      & --     & --       & 409.39   & (1.39)  \\
s35932            & 2315.5 & 987.09     & 18.08   & --      & 1536.89      & 57.91         & -- & 1495.87 & 1327.6   & 64.02        & (4.71)      & 703.03 & 1496.98  & 87.47    & (2.79)  \\
s38584            & --     & --         & 29.54   & --      & --           & --            & -- & --      & --       & 46.18        & (12.81)     & --     & --       & 58.54    & (9.91)  \\
s9234             & --     & --         & 1.13    & --      & --           & 3.7           & -- & --      & --       & 3.23         & (0.47)      & --     & --       & 33.33    & (0.13)  \\
systemcaes        & --     & --         & 7.38    & --      & --           & 148.06        & -- & --      & --       & --           & (22.18)     & --     & --       & 158.91   & (5.86)  \\
usb\_funct        & --     & --         & 77.41   & --      & --           & 172.28        & -- & --      & --       & 202.8        & (--)        & --     & --       & 3569.73  & (--)    \\
usb\_phy          & 52.35  & 26.85      & 0.2     & 2926.79 & 3073.82      & 0.59          & -- & 11.42   & 22.59    & 0.39         & (0.96)      & 194.26 & 2609.56  & 5.33     & (11.43) \\
wb\_conmax        & 815.48 & 1703.33    & 198.22  & --      & 20884.99     & 105.67        & -- & 775.02  & 1881.7   & 382.59       & (58.57)     & --     & 23900.31 & 1276.97  & (25.26) \\
\end{tabular}
\end{center}
\end{footnotesize}
\vspace*{-5mm}
\end{table*}

The results are reported in the columns whose headings do not contain `pipelined' in Table~\ref{tab:results_merge}.
Experiments denoted with a `--' are either timed out or ran out of memory.

Apparently, there are many infeasible benchmarks for rIC3. 

Our simple standalone standard IC3 is even less successful.
This is expected and in line with the results of, e.g., \cite{DBLP:conf/isvlsi/YuHNCKSCM18}. 

Our implementation \texttt{ic3-sim-F} that applies \emph{certified} most forward retiming
solves
\emph{all} benchmarks.

Our \texttt{ic3-sim} performs well on retimed \emph{and} combinationally 
resynthesized benchmarks, where rIC3 CPU-times predominantly increase. 
Additionally, even though certifying, it is not outperformed by ABC's 
\emph{non-certifying} \texttt{dprove}.
We especially refer to the non-trivial instances `\emph{usb\_funct}' and `\emph{b13}',
where our approach shows significant complementary strengths.

\paragraph{\textbf{Pipelined Designs}}\labelsec{ret_res_pipe}

We added pipelines of three FFs to each input of the instances from \reftab{results_merge}. We then applied the same retiming and resynthesis
operations like beforehand (\refsec{ret_res_std}). We refer to the columns labeled `pipelined' in Table~\ref{tab:results_merge}.

The added pipeline stages apparently hurt both rIC3 and \texttt{ic3}. Usually, pipelined and retimed
designs are a good fit for k-induction and we wanted to rule out that rIC3 
misses a lot of instances because simple path constraints are deactivated.
However, \emph{none} of the benchmarks could be solved by k-induction with simple-path
constraints (column `\textbf{rIC3-kind} simple-path') as well. 

Again, we present the results of \texttt{dprove} as a reference for retiming and resynthesis
and take note that we can solve two instances that the \emph{non-certifying} \texttt{dprove} cannot.

\paragraph{\textbf{Signal Correspondence With Virtual Forward Retiming}}\labelsec{scorr}
We showcase how our approach of creating virtually forward retimed (\refsec{virt_retiming}) signals on demand
enables significantly stronger preprocessing capabilities than ABC's widely used \texttt{scorr},
even though \texttt{scorr} \emph{uses unrolling}. 

Some results are displayed in \reftab{pipeline_scorr}. We allow an unrolling 
of the detected \emph{register depth} (\refsec{sse:sequentialsimulation}) 
$d$ of the designs (see column \emph{scorr -F d} in \reftab{pipeline_scorr})
as well as the maximum amount of frames that did not report `failed' induction 
checks due to exceeded SAT solver conflict budgets (see column \emph{scorr -F max} in 
\reftab{pipeline_scorr}).

We report on the fraction (`gain' in \%) of the found signal correspondences 
(i.e., the number of signals with another equivalent signal divided by the overall number of signals) 
on a subset of the resynthesized problems from \reftab{results_merge}.

One reason for the higher gains of our approach
might be that \texttt{scorr} switches off simple path constraints
for efficiency reasons. Example \refexample{simple_path} in \refsec{simple_path}
demonstrates that they can be necessary.

\begin{table}[tb]
\begin{footnotesize}
\begin{center}
\setlength{\tabcolsep}{2pt}
\caption{Signal correspondence . CPU time in $s$. Gain in \%.\labeltab{pipeline_scorr}}
\begin{tabular}{l|rr|rrr|rrr}
  & \multicolumn{2}{c|}{\emph{ic3-sim}} & \multicolumn{3}{c|}{\parbox{15mm}{\centering\emph{scorr -F d}}} & \multicolumn{3}{c}{\parbox{20mm}{\centering\emph{scorr -F max}}} \\
\hline
\parbox{14mm}{\textbf{Benchmark} \\ \textbf{pipelined}} & \textbf{time} & \textbf{gain} & \textbf{time}& \textit{d} & \textbf{gain} & \textbf{time} & \textit{max}& \textbf{gain}   \\
\hline
b13 & 1.22\,s & 100\,\%  & 0.1\,s &  7& 25.5\,\% & 11.06\,s & 100& 25.5\,\%   \\
b15 & 299.27\,s & 100\,\%   & 14.77\,s & 3& 6.16\,\% & 23.04\,s & 4 & 6.18\,\%  \\
s9234 & 33.33\,s & 100\,\%  & 7.3\,s& 7 & 51.26\,\%   & 8.35\,s& 20 & 51.26\,\%  \\
\end{tabular}
\end{center}
\end{footnotesize}
\vspace*{-5mm}
\end{table}

\paragraph{\textbf{Arbitrarily Strong Sequential Resynthesis}}
One important strength of our approach is its completeness. Even if no full
signal correspondence can be established, \texttt{ic3-sim} is able to strengthen
a proof of sequential equivalence with additionally learned lemmas. 
For this scenario, we increase the effort spent on finding good counterexamples for
refinement in \texttt{ic3-sim} by letting \emph{cutoff = 300}.

We deliberately chose `third party' sequential equivalence checking benchmarks of which we do not know 
anything about their transformation history. 
Firstly, in \reftab{results_hwmcc}, we demonstrate some hard examples from past Hardware Model Checking Competitions 
\cite{DBLP:journals/jsat/CabodiLPPPQVBH14} of which we know that the underlying problems are sequential miters. 

\begin{table}[tb]
\begin{footnotesize}
\begin{center}
\renewcommand{\arraystretch}{0.88}
\caption{Experimental results on HWMCC miters. CPU time in $s$.\labeltab{results_hwmcc}}
\begin{tabular}{l|rrr|rr}
\textbf{Benchmark} & \textbf{ic3} & \textbf{rIC3} & \textbf{ic3-sim} & {dprove} \textbf{(uncertified)} \\
\hline
bobsmhdlc & -- & 3098.61 & 44.13 & (--)\\
bobsmhdlc1 & -- & 2801.84 & 27.85 & (--)\\
bobsmhdlc2 & -- & 3546.76 & 42.42 & (--)\\
bobsmhdlc3 & -- & 3946.26 & 45.49 & (--)\\
bobsmcodic & -- & 7796.49 & 3175.13 & (--)\\ 
eijkbs3271 & -- & 17802.6 & 0.8 & (0.02)\\
eijkbs3384 & -- & -- & 78 & (25.57) \\
eijkbs6669 & -- & -- & 14.6 & (8.72) \\
\end{tabular}
\end{center}
\end{footnotesize}
\vspace*{-5mm}
\end{table}

The results show that we can perform well on problems that might not have complete retiming
invariants to discover. Our approach performs equally well on benchmarks
where the general purpose model checker rIC3 has difficulties and on benchmarks
where the dedicated equivalence checker \texttt{dprove} has difficulties.

Secondly, we present experiments on the 310 sequential 
equivalence checking \emph{6s22}-benchmarks from \cite{DBLP:conf/fmcad/DurejaGIV21}. This benchmark set was used to showcase how an IC3 implementation that
deduces inductive invariants over internal signals as well performs on realistic sequential equivalence checking problems.
\reftab{6s22results} shows our findings.
The IC3 implementation from \cite{DBLP:conf/fmcad/DurejaGIV21} is not openly available, we nevertheless copied the results that were reported with a timeout of 600\,$s$ in the paper into our table
and used a timeout of 600\,$s$ for our experiments as well.
Our implementation is \emph{able to solve {\bf all} benchmarks within 20\,$s$}, whereas rIC3 loses some instances
with the timeout of 600\,$s$ (wall clock time).
  
\begin{table}[tb]
\setlength{\tabcolsep}{2pt}
 \begin{footnotesize}
 \begin{center}
 \caption{Results on 6s22-SEC benchmarks from \cite{DBLP:conf/fmcad/DurejaGIV21}. \labeltab{6s22results}}
 \begin{tabular}{c|ccccc}
       &  \textbf{ic3-sim} & \textbf{IBM-IC3-INN}  \cite{DBLP:conf/fmcad/DurejaGIV21} & \textbf{ic3} & \textbf{rIC3} & \parbox{14mm}{\centering {dprove} \\ \textbf{(uncertified)}} \\
 \hline
      Solved & \textbf{310} / 310 & 278\footnotemark / 310 & 294 / 310 & 302 / 310 & (\textbf{310} / 310) \\
 \hline
 \end{tabular}
 \end{center}
 \end{footnotesize}
 \vspace*{-5mm}
\end{table}
\footnotetext{Not our measurement, taken from \cite{DBLP:conf/fmcad/DurejaGIV21}.}

\paragraph{\textbf{Certification Overhead}}\labelsec{certification}

\begin{table}[tb]
\begin{footnotesize}
\setlength{\tabcolsep}{3pt}
\centering
\caption{Certificate checking time using Certifaiger \cite{DBLP:conf/cav/FroleyksYPBH25}. \labeltab{certificates} }
\begin{tabular}{c|c|c|c|c}
   \parbox{0.8cm}{interval} & [0\,s, 1\,s) &  [1\,s, 10\,s) & [10\,s, 100\,s) & [100\,s, 1000\,s)   \\
\hline
\parbox{0.8cm}{bench\-marks}& \parbox{2cm}{b12, b13, s9234, usb\_phy, b12\_pipe, b13\_pipe, b14\_pipe, s9234\_pipe, usb\_phy\_pipe} & \parbox{2cm}{b14, b15, b17, b18, b20, b22, s15850\_pipe, s35932\_pipe, s38584\_pipe, s15850, s35932, s38584, aes\_core, mem\_ctrl, pci\_bridge32, des\_perf, wb\_conmax bobsmhdlc*, b18\_pipe, b20\_pipe, b21\_pipe, b22\_pipe, des\_perf\_pipe, mem\_ctrl\_pipe,  pci\_bridge32\_pipe} & \parbox{1.2cm}{b17\_pipe} & \parbox{2cm}{b15\_pipe, bobsmcodic} \\
\end{tabular}
\end{footnotesize}
\end{table} 

Here we report on the additional overhead that we experienced for certificate checking with Certifaiger.
We remark that certificate \emph{generation} comes naturally to our approach and only 
requires outputting the proven inductive invariant, which is negligible in
comparison to proving it beforehand with SAT solver calls.
For evaluating the certificate \emph{checking}, we considered the experiments with
 \texttt{ic3-sim-F} 
 (where the final invariant has to be transformed according to \refsec{retiming_preproc})
 and \texttt{ic3-sim}. We took the higher runtime if both 
 solved the benchmark.

The results are summarized in \reftab{certificates}.
Except for three outliers all certificates could be checked within less than 10\,s.
The highest run time was observed for the certificate for
bobsmcodic (output by \texttt{ic3-sim} with \emph{cutoff = 300}) that took us
449.5\,s to check.

\section{Conclusion and Future Work}
\labelsec{conclusion}

We presented a practical IC3-based approach for verifying sequential equivalence after retiming and complex sequential resynthesis.
By speculating invariants through (delayed) simulation, our method enables IC3 to efficiently prove equivalence for designs with retiming only, with retiming and combinational logic synthesis, and with 
retiming and strong sequential synthesis.

\paragraph*{Certified equivalence checking}
Our experiments demonstrate the need for a dedicated
sequential equivalence checker that emits
certificates. Apparently, using methods such as rIC3, although they implement techniques that
target equivalence-checking problems, is not sufficient.

\paragraph*{Non-certifying state-of-the-art}
We observe that even our prototype solution, which does not unroll the transition relation, has
complementary strengths compared to the non-certifying state of the art.
There are indeed instances on which our approach---although it is less tuned for efficiency than
ABC---performs significantly better.

Furthermore, we use virtual forward retiming only for signals that are identified as
candidates by multi-timeframe simulation. In this way, we can prove signal correspondences for which ABC's
sequential sweep techniques, which are based on $k$-induction without unique-state constraints, do not converge.

\paragraph*{Future work}
Directions for further research include moving more of the resynthesis techniques from
tools such as \texttt{dprove} into our tool and certifying them. For instance, tools typically
use passes of maximal forward retiming, minimal register retiming, and
combinational rewriting attempts interchangeably to align FFs.

The theory presented in this paper already appears to enable certification of such approaches.

\bibliographystyle{IEEEtran}
\bibliography{paper}

\end{document}